%% file: compressive_decomposition_arxiv.tex
\documentclass{article}

\input{jw_defs.tex}

\numberwithin{equation}{section}
\begin{document}

\title{Compressive Principal Component Pursuit}
\author{John Wright$^1$, Arvind Ganesh$^2$, Kerui Min$^2$, and Yi Ma$^{2,3}$ \vspace{3mm} \\
$^1$ Electrical Engineering Department, Columbia University, New York  \vspace{2mm}\\
$^2$ Electrical and Computer Engineering Department, UIUC, Urbana \vspace{2mm}\\
$^3$ Microsoft Research Asia, Beijing, China
}
\maketitle

\begin{abstract} We consider the problem of recovering a target matrix that is a superposition of low-rank and sparse components, from a small set of linear measurements. This problem arises in compressed sensing of structured high-dimensional signals such as videos and hyperspectral images, as well as in the analysis of transformation invariant low-rank recovery. We analyze the performance of the natural convex heuristic for solving this problem, under the assumption that measurements are chosen uniformly at random. We prove that this heuristic exactly recovers low-rank and sparse terms, provided the number of observations exceeds the number of intrinsic degrees of freedom of the component signals by a polylogarithmic factor. Our analysis introduces several ideas that may be of independent interest for the more general problem of compressed sensing and decomposing superpositions of multiple structured signals. 
\end{abstract}


\section{Introduction}
In recent years, there has been tremendous interest in recovering low-dimensional structure in high-dimensional signal or data spaces. This interest has been fueled by the discovery that efficient techniques based on convex programming can accurately recover low-complexity signals such as sparse vectors or low-rank matrices from severely compressive, incomplete, or even corrupted observations. 

One representative example arises in {\em Robust Principal Component Analysis (RPCA)}. There, the goal is to recover a low-rank matrix $\mb L_0$ from grossly corrupted observations. For example, suppose we observe $\mb M = \mb L_0 + \mb S_0$, where $\mb S_0$ is a sparse error. Under mild conditions, the following convex program, called {\em Principal Component Pursuit} (PCP) \cite{Candes2011-JACM,Chandrasekharan2011-SJO}:
\begin{equation}
\minimize{\| \mb L\|_* + \lambda \| \mb S \|_1}{\mb L + \mb S = \mb M},
\label{eqn:pcp}
\end{equation}
precisely recovers $\mb L_0$ and $\mb S_0$. In \eqref{eqn:pcp}, $\| \cdot \|_*$ is the matrix {\em nuclear norm} (sum of singular values) and $\| \cdot \|_1$ is the $\ell^1$ norm (sum of magnitudes). For data analysis applications, this suggests that a low-rank matrix $\mb L_0$ can be recovered from the observation $\mb M$ despite large-magnitude sparse errors. This result has been extended and generalized in a number of directions: to include additional small dense noise $\mb M = \mb L_0 + \mb S_0 + \mb N$ \cite{Zhou2010-ISIT}, large fractions of random errors $\mb S_0$ \cite{Ganesh2010-ISIT}, and even column-sparse or row-sparse errors \cite{Xu2011-IT,McCoy2011-EJS}. 

The conditions under which recovery is known to occur are fairly broad: provided the low-rank term satisfies a technical {\em incoherence} condition, correct recovery can occur even when $\mathrm{rank}(\mb L_0)$ almost proportional to dimension of the matrix $\mb M$, and the number of nonzero entries in $\mb S_0$ is proportional to the number of entries in $\mb M$ \cite{Candes2011-JACM}. On the other hand, in many applications of interest, the rank may actually be significantly smaller than dimension (say 3 \cite{Wu2010-ACCV}, or 9 \cite{Basri2003-PAMI}). Moreover, cardinality of the sparse term may also be quite small. In such a situation our number $mn$ of observations could be extravagantly large compared to the number degrees of freedom in the unknowns $\mb L_0,\mb S_0$. Is it possible to recover $\mb L_0$ and $\mb S_0$ from smaller sets of linear measurements? 


\subsection{Compressive RPCA}

The low-rank and sparse model described above captures properties of many signals of interest, including foreground and background in video surveillance \cite{Candes2011-JACM}, videos \cite{Ganesh2011-PAMI,Shu2011-ICCV}, structured textures \cite{Zhang2011-IJCV}, hyperspectral datacubes \cite{Waters2011-NIPS,Golbabaee2011-ICASSP} and more. The ability to recover low-rank and sparse models from small sets of linear measurements could be very useful for developing new sensing architectures for such signals \cite{Donoho2006-IT,Waters2011-NIPS}. Mathematically, our observations have the form 
\begin{equation}
\mb D \;\doteq\; \mc P_Q[ \mb M ] \;=\; \PQ[ \mb L_0 + \mb S_0 ],
\end{equation}
where $Q \subseteq \Re^{m \times n}$ is a linear subspace, and $\PQ$ denotes the projection operator onto that subspace. Can we simultaneously recover the low-rank and sparse components correctly from highly compressive measurements via the natural convex program
\begin{equation} \label{eqn:cpcp-intro}
\minimize{\|\mb L\|_* + \lambda \|\mb S\|_1}{\PQ[\mb L + \mb S] = \mb D} \;?
\end{equation}
While this question is largely open, there is good reason to believe the answer may be positive. For example, \cite{Candes2011-JACM,Li2011-pp} have studied the ``robust matrix completion'' problem, with $\PQ = \mc P_{\Omega}$, where $\Omega$ is a small subset of the entries of the matrix. When $\PQ = \mc P_{\Omega}$, it is impossible to exactly recover $\mb S_0$ (many of the entries are simply not observed!), but the low-rank term $\mb L_0$ {\em can} be recovered from near-minimal sets of random samples \cite{Li2011-pp}. However, in many applications the sparse term $\mb S_0$ is actually the quantity of interest: for example, in visual surveillance, $\mb S_0$ might capture moving foreground objects. To recover both $\mb L_0$ and $\mb S_0$, we must require measurements $Q$ that are incoherent with {\em both} the low-rank and the sparse term.

\begin{figure*}[ht]
\centerline{
    \subfigure[$p = 0$]{
        \includegraphics[height=0.26\columnwidth]{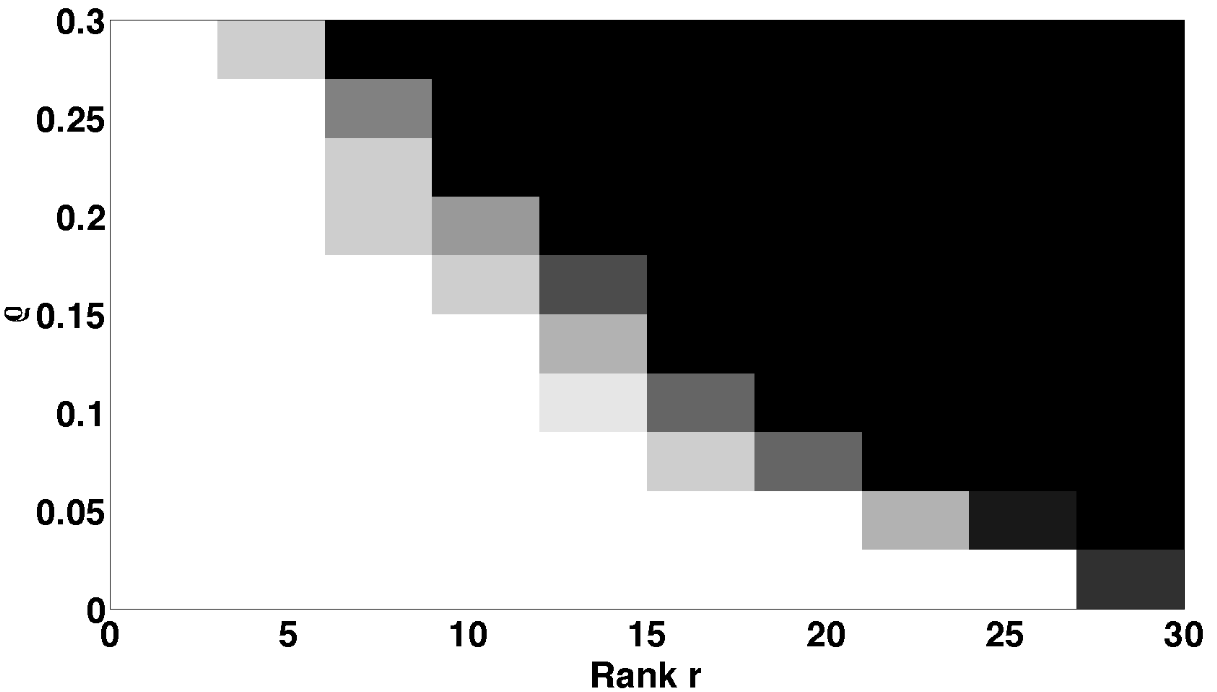}
    }
    \subfigure[$p = 500$]{
        \includegraphics[height=0.26\columnwidth]{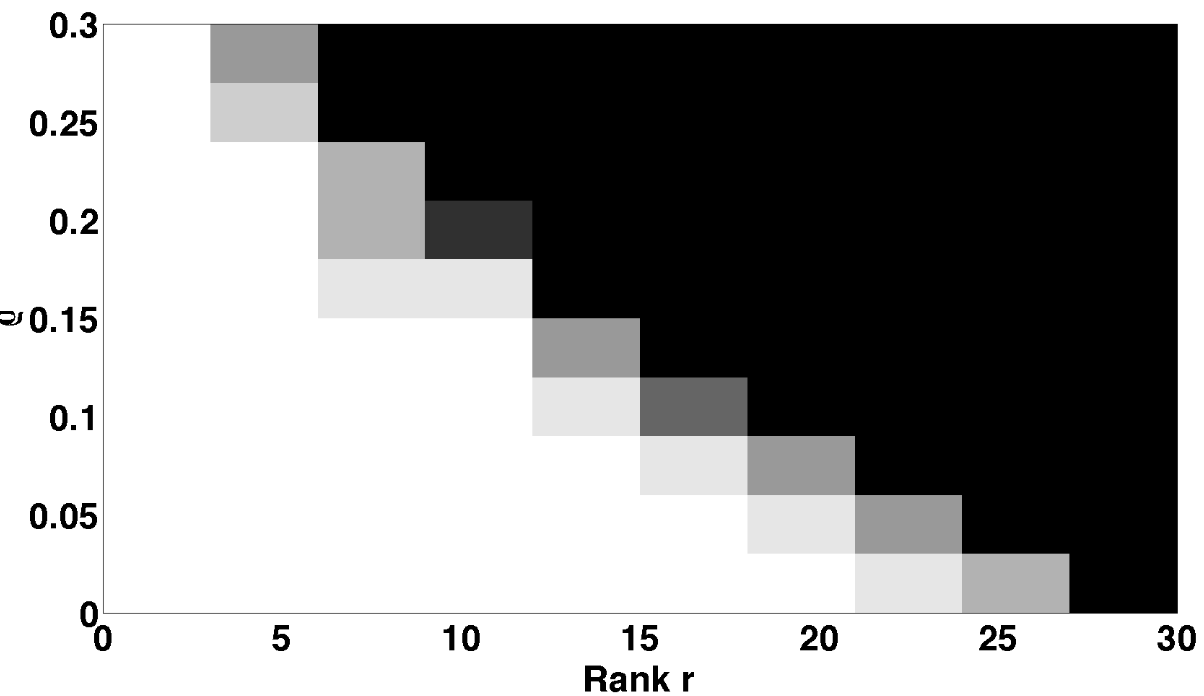}
    }
}    
\centerline{    
    \subfigure[$p = 2,000$]{
        \includegraphics[height=0.26\columnwidth]{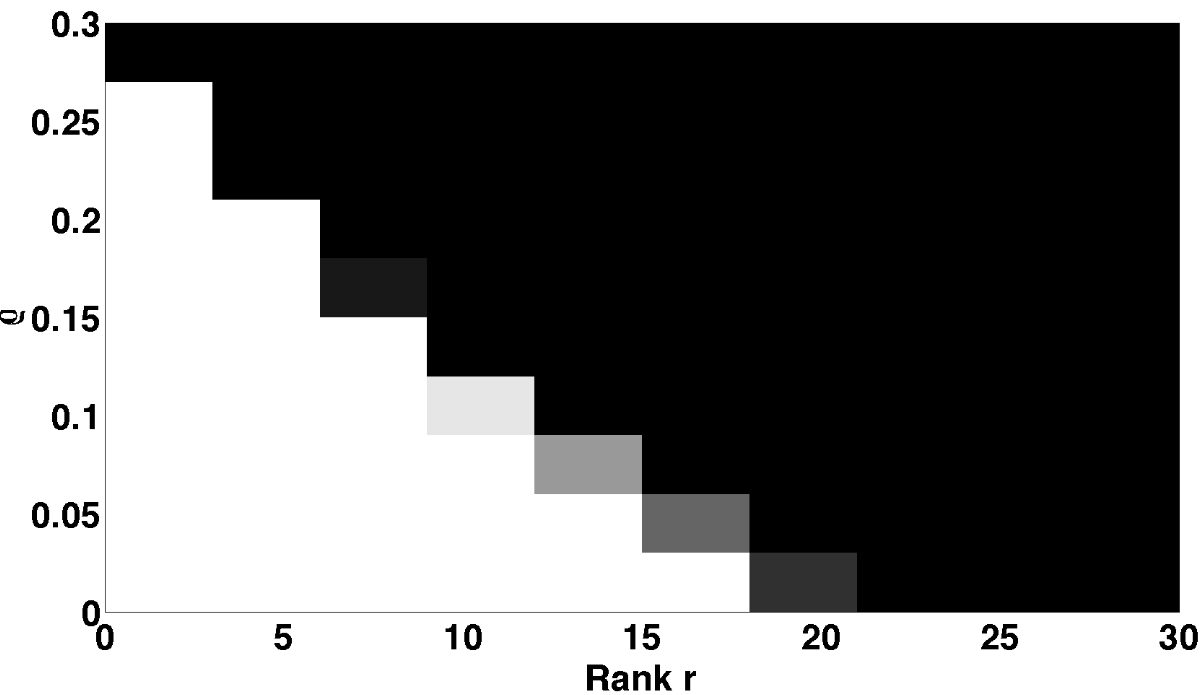}
    }
    \subfigure[$p = 5,000$]{
        \includegraphics[height=0.26\columnwidth]{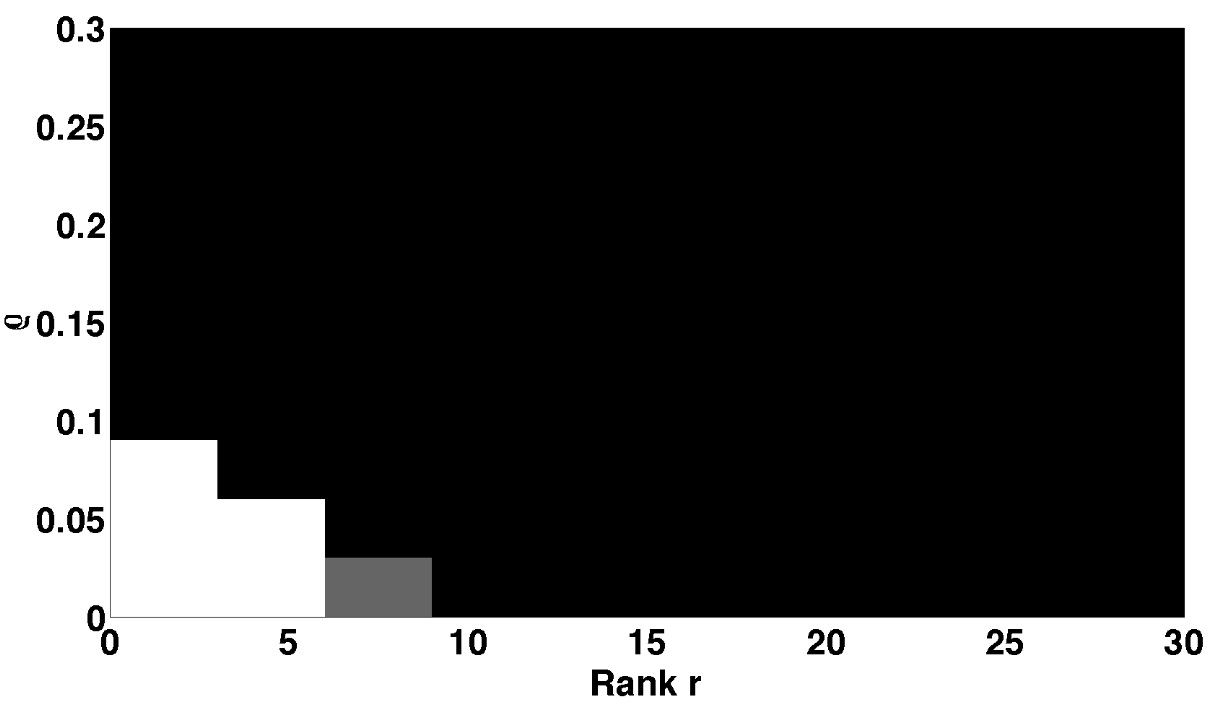}
    }
}
\caption{{\bf Compressive sensing of low-rank and sparse matrices via convex program \eqref{eqn:cpcp-intro}.} We solve the convex program for recovering an $m\times m$ matrix $\mb M  = \mb L_0 + \mb S_0$ with $m = 100$ from $q = m^2 - p$ random linear measurements $\PQ[\mb M]$. For each subplot: the $x$-axis is the rank $r$ of the matrix $\mb L_0$ and the $y$-axis is the percentage of non-zero entries in $\mb S_0$. The intensity is proportional to the probability of success with pure white color meaning 100\% (out of 10 random trials). Notice that when $p = 5,000$, the number of linear measurements is only half of the number of entries and there remains a small region where the convex program succeeds. These simulations use an accelerated gradient algorithm with continuation, similar to \cite{Becker2010-pp}.}
\label{fig:simulations}
\end{figure*}

In this paper, we investigate the performance of \eqref{eqn:cpcp-intro} when $Q$ is a randomly chosen subspace, (incoherent with $\mb L_0$ and $\mb S_0$ with high probability). As the simulation results in Figure \ref{fig:simulations} suggest, as long as the rank and sparsity are low enough, we can expect the convex program to correctly recover both the low-rank and sparse components from a reduced set of random linear measurements. A similar recovery problem was recently considered by \cite{Waters2011-NIPS}, again, with the goal of designing sensing strategies capable of recovering both $\mb L_0$ and $\mb S_0$. We will discuss the results of \cite{Waters2011-NIPS} and other related works in more detail in Section \ref{sec:relation}, after we have stated our main result.

\subsection{Transformed RPCA}

Aside from the perspective of compressive sensing, there are many other practical scenarios that require recovering a low-rank matrix from partial, incomplete, or corrupted measurements. One example is when the given data is a transformed version of the low-rank and sparse matrices: 
\begin{equation} \label{eqn:txfm}
\mb M \circ \tau = \mb L_0 + \mb S_0,
\end{equation}
 where $\tau$ is an unknown nonlinear transformation from some continuous group $\mathcal G$. The goal is to simultaneously recover $\mb L_0, \mb S_0$ and $\tau$ from $\mb M$. One can view this as a ``transformed RPCA'' problem. The constraint \eqref{eqn:txfm} is often highly nonlinear. One popular approach is to linearize the measurements against parameters of the transformation:
$$ \mb M \circ \tau + \mc J  [ \Delta \tau ] \approx \mb L_0 + \mb S_0,$$
where $\mc J$ is the Jacobian of $\mb M\circ \tau$ against of $\tau$. We can then solve for an increment $\Delta \tau = \tau_{k+1} - \tau_k$ in the transformation parameters via the convex program:
\begin{equation}
\text{minimize}_{\mb L,\mb S,\Delta \tau} \,\; \|\mb L\|_* + \lambda \|\mb S\|_1 \quad \text{subject to} \quad \mb M \circ \tau_k + \mc J[ \Delta \tau ] = \mb L + \mb S.
\end{equation}
Mathematically, this program is equivalent to \eqref{eqn:cpcp-intro}. To see this, let $Q$ be the orthogonal complement to the range of $\mc J$, so that $\PQ \mc J = 0$. Let $$\mb D  \;\doteq\; \PQ [\mb M \circ \tau+ \mc J [\Delta \tau] ] \;=\; \PQ [\mb M \circ \tau] \;\approx\; \PQ [\mb L_0 + \mb S_0].$$ After $\Delta \tau$ is eliminated in this way, the problem now becomes recovering the low-rank and sparse components from $\mb D$: 
\begin{equation}
\text{minimize}_{\mb L,\mb S} \,\; \|\mb L\|_* + \lambda \|\mb S\|_1 \quad \text{subject to} \quad \mb D = \PQ[\mb L + \mb S],
\end{equation}
which has the same form as above. 

Empirically, this iterative linearization scheme performs well in  applications such as aligning multiple images \cite{Ganesh2011-PAMI} and rectifying low-rank textures \cite{Zhang2011-IJCV}. Hence, it is important to understand under what conditions we should expect the associated convex program to perform correctly. However, there are some important differences from the compressive sensing scenario:
\begin{enumerate}
\item In the transformed RPCA case, we often are dealing with a finite dimensional deformation group $\mathcal G$ whose dimension, say $p$, is either fixed (as in \cite{Zhang2011-IJCV}) or grows very slowly compared to the number of entries in the matrix (as in \cite{Ganesh2011-PAMI}). 
\item Unlike compressive sensing where the measurement operator $\PQ(\cdot)$ can be arbitrarily chosen, here it is determined by the given data and the associated transformation group. We can no longer model it as a random projection. Hence, we hope to have deterministic conditions which can be directly verified with the given data.
\end{enumerate}

\subsection{Compressive Sensing of Decomposable Components}
From both the compressive and transformed RPCA problems, we see the need to understand under what conditions we should expect to correctly recover the low-rank and sparse components from compressive or partial measurements: $\mb D = \PQ[\mb L_0 + \mb S_0]$. In particular, we are interested in when the convex program:
\begin{equation} \label{eqn:cpcp}
\text{minimize} \,\; \|\mb L\|_* + \lambda \|\mb S\|_1 \quad \text{subject to} \quad \mb D = \PQ[\mb L + \mb S],
\end{equation}
finds the correct solution $\mb L_0$ and $\mb S_0$. Following the terminology of \cite{Candes2011-JACM}, in this paper we refer to this convex program as {\em Compressive Principal Component Pursuit} (CPCP). 

One fundamental question is how many measurements $q$ are needed for the above program \eqref{eqn:cpcp} to correctly recover $\mb L_0$ and $\mb S_0$.  Clearly, this number should be bounded from below by the number of intrinsic degrees of freedom in $(\mb L_0,\mb S_0)$. Since a rank $r$ matrix has $(m+n-r)r$ degrees of freedom, the number of continuous degrees of freedom in the pair $(\mb L_0,\mb S_0)$ is equal to 
$$(m + n - r) r + \norm{\mb S_0}{0},$$
where we recall that $\norm{ \,\cdot\, }{0}$ denotes the number of nonzero entries in a matrix. Hence, the best we can possibly hope for is a number of measurements $q$ on this order. We will show that when the measurements are random (say Gaussian), the desired $(\mb L_0,\mb S_0)$ can indeed be exactly recovered from a number of measurements that is very close to this lower bound: provided 
$$\mathtt{\# measurements} \;\ge\; O(\log^2 m) \times \mathtt{\# degrees\; of\; freedom}( \mb L_0, \mb S_0  ), $$ 
the compressive principal component pursuit program  \eqref{eqn:cpcp} correctly recovers this pair with very high probability. Notice that this bound is nearly optimal, differing from the hard lower bound by only a polylogarithmic factor. 

Our analysis actually pertains to a much more general class of problems of decomposing a given observation into multiple incoherent components:
\begin{equation} \label{eqn:general-inco}
\minimize{\sum_i \lambda_i \norm{\mb X_i}{(i)}}{ \sum_i \mb X_i = \mb M }.
\end{equation}
Here, $\| \cdot\|_{(i)}$ are (decomposable) norms that encourage various types of low-complexity structure. Principal Component Pursuit \cite{Candes2011-JACM,Chandrasekharan2011-SJO}, Outlier Pursuit \cite{Xu2011-IT,McCoy2011-EJS} and Morphological Component Analysis \cite{Bobin2006-TSP} are all special cases of this general problem. Roughly speaking, our analysis will suggest that, if the above program succeeds in recovering all the components $\{\mb X_i\}$ from $\mb M$, one should also expect to recover them from the highly compressive measurements $\PQ[\mb M]$. The number of measurements required is again governed by the intrinsic degrees of freedom $\{\mb X_i \}$ multiplying at most a $\mbox{polylog}(m)$ factor. Thus, we believe results in this paper are not only limited to decomposing low-rank and sparse signals but also applicable to a broad class of source separation or signal decomposition problems that may arise in signal processing, communications, and pattern recognition.

The remainder of this paper is organized as follows. In Section \ref{sec:model-result}, we first introduce the precise mathematical model and present the main technical results of this paper. In Section \ref{sec:relation}, we discuss its implications and relationships with existing work in the literature. Section \ref{sec:upgrade} discusses the more general setting of \eqref{eqn:general-inco} and lays out the framework of our analysis. The remaining sections complete the proof of our main results.

\section{Models and Main Results}\label{sec:model-result}

Our main technical contribution is a procedure for producing a certificate of optimality for $(\mb L_0,\mb S_0)$ for the Compressive Principal Component Pursuit problem, given that the pair is optimal for Principal Component Pursuit. In this sense, our mathematical approach is modular -- it partially decouples the analysis of the the compressive measurements from the analysis of the core low-rank and sparse recovery problem. Combining with existing models and analyses of PCP, we can prove that the pair $(\mb L_0,\mb S_0)$ is indeed recoverable by the convex optimization. 

We first recall conditions under which $\mb M = \mb L_0 + \mb S_0$ can be exactly separated into its constituents, by PCP. Intuitively, we should not expect to recover all possible low-rank pairs and sparse pairs $(\mb L_0,\mb S_0)$. Indeed, imagine the case when $\mb M$ is rank-one and one-sparse (i.e., $\mb M = \mb e_i \mb e_j^*$ for some $i,j$). In this situation the answers $(\mb L = \mb e_i\mb e_j^*,\mb S =\mb 0)$ and $(\mb L = \mb 0,\mb S = \mb e_i\mb e_j^*)$ both seem reasonable -- the decomposition problem is ambiguous! 

To make the problem meaningful, we need conditions that ensure that (i) the low-rank term $\mb L_0$ does not ``look sparse'' and (ii) the sparse term $\mb S_0$ does not ``look low-rank.'' One popular way formalizing the first intuition of doing this is via the notion of {\em incoherence} introduced by \cite{Candes2008}. If the low-rank matrix $\mb L_0$ has rank-reduced singular value decomposition $\mb L_0 = \mb U \mb \Sigma \mb V^*$, then we say that $\mb L_0$ is $\mu$-incoherent if 
\begin{eqnarray}
\forall \, i \; \norm{\mb U^* \mb e_i}{2}^2 \;\le\;\frac{\mu r}{m}, \quad \forall \, j  \; \norm{\mb V^* \mb e_j}{2}^2 \;\le\;\frac{\mu r}{n}, \quad \text{and} \quad \norm{\UVt}{\infty} \;\le\; \sqrt{\frac{\mu r}{mn}}.
\end{eqnarray}
Intuitively, these conditions ensure that the singular vectors of $\mb L_0$ are not too concentrated on only a few coordinates -- the singular vectors do not ``look sparse.'' For further discussion of the implications of this condition, we refer the reader to \cite{Candes2008}.

At the same time, we need to ensure that the sparse term does not ``look low-rank.'' One appealing way of doing this is via a random model: we assume that each $(i,j)$ is an element of $\supp{\mb S_0}$ independently with probability $\rho$ bounded by some small constant. We assume that the signs of the nonzero entries are independent symmetric $\pm 1$ random variables (i.e., Rademacher random variables). In stating our theorems, we call such a distribution an ``iid Bernoulli-Rademacher model.'' 

Thus far, we have discussed only the low-rank and sparse terms, but not the properties of the measurements $Q$. We will give a result for the case when $Q$ is a chosen uniformly at random from the set of all $q$-dimensional subspaces of $\Re^{m \times n}$. More precisely, $Q$ is distributed according to the Haar measure on the Grassmannian $\bb G(\Re^{m \times n},q)$. This means that the distribution of $Q$ is rotationally invariant. On a more intuitive level, this means that $Q$ is equal in distribution to the linear span of a collection of $q$ independent iid $\normal{0}{1}$ matrices. In notation more familiar from compressed sensing, we may let $\mb Q_1, \dots, \mb Q_q$ denote such a set of matrices, and define an operator $\mc Q : \Re^{m \times n} \to \Re^q$ via 
\begin{equation}
\mc Q[ \mb M]\; = \; \left(\innerprod{\mb Q_1}{\mb M}, \dots, \innerprod{\mb Q_q}{\mb M}\right)^* \;\in\; \Re^q.
\end{equation}
Our analysis also pertains to the equivalent convex program: 
\begin{equation} \label{eqn:proj-form}
\minimize{\norm{\mb L}{*} + \lambda \norm{\mb S}{1}}{\mc Q[ \mb L + \mb S ] = \mc Q[ \mb L_0 + \mb S_0 ].}
\end{equation}
Indeed, since $\mc Q$ has full rank $q$ almost surely, \eqref{eqn:proj-form} and \eqref{eqn:cpcp} are completely equivalent. 

Under this setting, the following theorem gives a tight bound on the number of (random) measurements required to correctly recover the pair $(\mb L_0, \mb S_0)$ from $\PQ[\mb M]$ via CPCP:
\begin{theorem}[\bf Compressive PCP Recovery] \label{thm:main} Let $\mb L_0,\mb S_0 \in \Re^{m \times n}$, with $m \ge n$, and suppose that $\mb L_0 \ne \mb 0$ is a rank-$r$, $\mu$-incoherent matrix with
\begin{equation}
r \;\le\; \frac{c_r n}{ \mu \log^2 m },
\end{equation}
and $\sign{\mb S_0}$ is iid Bernoulli-Rademacher with nonzero probability $\rho < c_\rho$. Let $Q \subset \Re^{m \times n}$ be a random subspace of dimension 
\begin{equation}
\dim{Q} \;\ge\; C_Q \cdot (  \rho m n + m r ) \cdot \log^2 m
\end{equation}
distributed according to the Haar measure, probabilistically independent of $\mathrm{sign}(\mb S_0)$. Then  with probability at least $1 - C m^{-9}$ in $(\mathrm{sign}(\mb S_0),Q)$, the solution to 
\begin{equation}
\minimize{\norm{\mb L}{*} + \lambda \norm{\mb S}{1}}{\PQ [ \mb L + \mb S ] =  \PQ[\mb L_0 + \mb S_0]}
\end{equation}
with $\lambda = 1/ \sqrt{m}$ is unique, and equal to $(\mb L_0,\mb S_0)$. Above, $c_r,c_\rho,C_Q,C$ are positive numerical constants. 
\end{theorem}
\vspace{.25in}

Here, the magnitudes of the nonzeros in $\mb S_0$ are arbitrary, and no randomness is assumed in $\mb L_0$. The randomness in our main result is in the sign and support pattern of $\mb S_0$ and the measurements $Q$. We note in passing that the randomness in the signs of $\mb S_0$ can be removed using the techniques of \cite{Candes2011-JACM} Sections 2.1-2.2. We will not pursue this in great depth here. 

We also note that the bounds on $r$ and $\rho$ essentially match those of \cite{Candes2011-JACM}, possibly with different constants. So, again, $r$ and $\| \mb S_0 \|_0$ can be rather large. On the other hand, when these quantities are small, the bound on $\dim{Q}$ ensures that the number of measurements needed for accurate recovery is also commensurately small. We will compare our results to other works from the literature in the next section. First, we pause to fix some notation.

\paragraph{Notation.} Bold uppercase letters $\mb A,\mb B,\dots$ denote matrices. Bold lowercase letters $\mb x,\mb y$ denote vectors. Script uppercase letters $\mc A,\mc B,\dots$ denote operators on matrices. In particular, if $S \subset \Re^{m \times n}$ is a linear subspace, we will let $\PS$ denote the orthogonal projection onto $S$. The notations $C,c$ will always refer to numerical constants. When used in different sections they may not refer to the same constant. All logarithms are base-$e$. ``$\oplus$'' denotes a direct sum between linearly independent subspaces. When applied to subsets of a vector space, ``+'' will denote Minkowski summation, i.e., $A+B = \set{\mb a + \mb b \mid \mb a \in A,\;\mb b \in B }$. 
\begin{definition}
We will say that subspaces $S_1, \dots, S_k$ are {\bf \em independent} if $$\dim{ S_1 + \dots + S_k } \;=\; \dim{ S_1 } + \dots + \dim{S_k}.$$
\end{definition}


\section{Relationship to the Literature}\label{sec:relation}
As mentioned above, in recent years there has been a large amount of work on matrix recovery and decomposition, for example see \cite{Candes2011-JACM,Chandrasekharan2011-SJO,Zhou2010-ISIT,Ganesh2010-ISIT,Xu2011-IT,McCoy2011-EJS,Agarwal2011-pp,Hsu2011-IT} and references therein. The aforementioned works all pertain to the case when the matrix $\mb M$ is fully observed, and hence are not directly comparable to our result. In Section \ref{sec:upgrade}, we will see that our analysis gives a tool for transforming a certificate of optimality for the fully observed problem into a certificate of optimality for the compressive problem. Because this technique is modular, it may be possible to apply it in conjunction with the aforementioned works to prove correct recovery under different assumptions, and even with different regularizers.

Compared to the fully observed problem, there is much less dedicated work on low-rank and sparse recovery from compressive measurements. Recently, motivated by applications in compressive foreground and background separation and compressive hyperspectral image acquisition, \cite{Waters2011-NIPS} introduced a greedy algorithm for this problem, which aims at the objective function
\begin{equation}
\text{minimize}_{\mb L,\mb S} \;\; \|\mb D - \PQ[\mb L + \mb S]\|_2 \quad \mathrm{subject \; to} \quad \rank{\mb L} \le r,\; \|\mb S\|_0 \le k.
\end{equation}
Their algorithm is similar in spirit to the CoSaMP algorithm of \cite{NeedellD2008} for recovering sparse signals, and performs well on numerical examples. Analyzing its behavior theoretically and proving performance guarantees is currently an open problem. 

As the body of results on specific problems such as matrix recovery grows, there has been an increasing interest in unifying or generalizing the basic insights obtained from studying special cases. A number of groups have produced results that pertain to general structured regularizers. For example, Negahban et.\ al. \cite{Negahban2010-pp} have introduced a general geometric framework for analyzing low-complexity signal recovery, highlighting the role of the regularizer in overcoming a lack of strong convexity in the loss. Agarwal et.\ al.\  \cite{Agarwal2011-pp} use this framework to analyze sparse and low-rank decomposition, and have obtained tight results for estimation in noise, stronger than previously known results by \cite{Zhou2010-ISIT}. Their analysis proceeds under different (weaker) assumptions, which preclude exact recovery. 

In a similar vein, Chandrasekaran et.\ al.\ \cite{Chandrasekaran2010-pp} have recently produced a very general analysis of structured signal recovery with Gaussian measurements. That work exploits the geometry of the atomic norm ball -- in particular, relating the required number of measurements to the Gaussian width of the tangent cone at the desired solution. Based on this, they give tight bounds on the number of measurements needed to recover a low-rank matrix or sparse vector. However, once the atomic set contains both low-rank and sparse matrices, it is less clear how to analyze the Gaussian width of the tangent cone. Indeed, the non-trivial analysis in \cite{Candes2011-JACM,Chandrasekharan2011-SJO} can be viewed as simply showing that the desired solution lies on the boundary of the norm ball. Estimating the width of the tangent cone at that point seems to entail additional analytical difficulty. 

For Gaussian measurements, the recent work of Cand\`{e}s and Recht  \cite{Candes2011-pp} also gives simple bounds for exact recovery, under the assumption that the regularizer (or norm) is {\em decomposable}. If we wished to apply similar analysis to our problem, we would need to work with the quotient norm on $\mb M$:
\begin{equation}
\| \mb M \|_\diamond \doteq \inf_{\mb L + \mb S = \mb M} \| \mb L \|_* + \lambda \| \mb S \|_1.
\end{equation}
This is the {\em infimal convolution} of two decomposable terms. Its subdifferential has a number of nice properties which we will exploit in our analysis, but decomposability (in the sense of \cite{Candes2011-pp}) does not appear to be one of them. Nevertheless, the results in this paper show that under suitable conditions, we should expect the same type of compressive sensing results for this class of generalized norms for superpositions of low-complexity components. 

In this paper, we generalize the analysis of decomposable regularizers to their sums (or strictly speaking infimal convolutions) and obtain nearly optimal bounds on the required number of measurements for exact recovery and decomposition of low-complexity components via convex optimization. In particular, our results provide strong theoretical justification for conducting robust principal component analysis with highly compressive measurements. Because our results assume a random model for the operator $\mc Q$, of the two application scenarios described in the introduction, our results are likely to be more applicable to the compressive sampling scenario. Indeed -- the challenge for analyzing transformed matrix recovery problems seems not to lie in elucidating the absolutely minimum number of ``measurements'', but rather in dealing with dependencies between the operator $\mc Q$ and the solution of interest $(\mb L_0,\mb S_0)$. In a companion paper \cite{Ganesh2012}, we give results for deterministic operators $\mc Q$, which may be applicable to that situation.

\section{General Certificate Upgrades} \label{sec:upgrade}

In this section, we present the technical result used to obtain Theorem \ref{thm:main} above. As promised, this result will have implications for compressive variants of a large number of conceivable signal decomposition problems. In full generality, we can imagine that the fully observed data $\mb M$ are given as a sum of structured terms: 
\begin{equation}
\mb M \;=\; \mb X_1 + \mb X_2 + \dots + \mb X_\tau,
\end{equation}
where each $\mb X_i$ satisfies a low-complexity model such as sparsity or rank-deficiency, possibly also including more exotic types of structured sparsity \cite{Bach2010-pp}. For each type of structure, we have a corresponding regularizer $\| \cdot \|_{(i)}$. The natural convex heuristic for decomposing $\mb M$ into its components would solve 
\begin{equation} \label{eqn:general}
\minimize{\sum_i \lambda_i \norm{\mb X_i}{(i)}}{ \sum_i \mb X_i = \mb M },
\end{equation}
where the $\lambda_i > 0$ are scalar weight factors. Many authors have studied special cases of this problem, and given conditions under which correct decomposition occurs. A prime example is Principal Component Pursuit; others include Outlier Pursuit \cite{Xu2011-IT,McCoy2011-EJS} and Morphological Component Analysis \cite{Bobin2006-TSP}. 

The goal of this paper is not to study \eqref{eqn:general} per se, but rather to understand what happens to it when we only observe compressive measurements of $\mb M$ (or when $\mb M$ itself is subject to some transformation): 
\begin{equation} \label{eqn:general-c}
\minimize{\sum_i \lambda_i \norm{\mb X_i}{(i)}}{ \PQ \left[\sum_i \mb X_i\right] = \PQ \mb M }.
\end{equation}
Suppose we know that \eqref{eqn:general} correctly decomposes $\mb M$ into $\mb X_1, \dots, \mb X_\tau$. Does this imply that \eqref{eqn:general-c} can also recover $\mb X_1, \dots, \mb X_\tau$? At a slightly more technical level, we can ask whether a certificate of optimality for the decomposition problem \eqref{eqn:general} can be refined to also certify optimality for the compressive decomposition problem \eqref{eqn:general-c}. Theorem \ref{thm:upgrade} below will imply that this is true under broad circumstances. Provided we have proved optimality for \eqref{eqn:general}, we can move to optimality for \eqref{eqn:general-c}, as long as the number of measurements $\dim{Q}$ is sufficiently large. In this sense, our analysis is modular: any technique can be used to perform the analysis of the original decomposition problem, provided it constructs an (approximate) dual certificate. 

\paragraph{Duality and Optimality.} 

Our result pertains to {\em decomposable} norms $\|\cdot \|_{(i)}$ \cite{Negahban2010-pp,Candes2011-pp}. This notion includes many sparsity inducing norms, such as the $\ell^1$ norm and nuclear norm (as above), as well as sums of block $\ell^p$ norms. 

\begin{definition} \label{def:decomp} We say that a norm $\|\cdot\|$ is {\bf \em decomposable} at $\mb X$ if there exists a subspace $T$ and a matrix $\mb S$ such that 
\begin{equation}
\partial \| \,\cdot\, \|(\mb X) \;=\; \set{ \, \mb \Lambda \mid \mc P_{T} \mb \Lambda \,=\, \mb S, \; 
\| \mc P_{T^\perp} \mb \Lambda \|^* \,\le\, 1 },
\end{equation}
where $\| \cdot \|^*$ denotes the dual norm of $\| \cdot \|$, and $\PTp$ is nonexpansive with respect to $\| \cdot \|^*$.
\end{definition}
For example, the $\ell^1$ norm satisfies this definition with $T = \supp{\mb X}$ and $\mb S = \sign{\mb X}$. The above definition is completely equivalent to that of \cite{Candes2011-pp}. It is also related to the definition of \cite{Negahban2010-pp}, but not strictly equivalent to it.\footnote{To be clear, in the sense of \cite{Negahban2010-pp},  a norm $\| \cdot \|$ is decomposable over a subspace pair $T,T^\perp$ if for all $\mb x \in T$, $\mb y \in T^\perp$, $\| \mb x + \mb y \| = \| \mb x \|+ \|\mb y \|$. If $\mb x \in T$, and $\| \cdot \|$ is decomposable over $T,T^\perp$ in the sense of \cite{Negahban2010-pp}, and the restriction of $\| \cdot \|$ to $T$ is differentiable at $\mb x$, then it is decomposable in the sense of Definition \ref{def:decomp}. On the other hand, norms that are decomposable in the sense of Definition \ref{def:decomp} need not be decomposable in the sense of \cite{Negahban2010-pp}, and so the two notions are not strictly comparable.} We assume that each $\| \cdot \|_{(i)}$ is decomposable at the target solution $\mb X_{i,\star}$, so per the above definition we have a sequence of subspaces $T_i$ and matrices $\mb S_i$ that define the subdifferentials of each of the regularizers $\| \cdot \|_{(i)}$. With this notation in mind, we can state a simple sufficient optimality condition for \eqref{eqn:general-c}:
\begin{lemma} \label{lem:optimality} Consider a feasible solution $\mb x_\star = (\mb X_{1,\star}, \dots, \mb X_{\tau,\star})$ to \eqref{eqn:general-c}. Suppose that each of the norms $\| \cdot \|_{(i)}$ is decomposable at $\mb X_{i,\star}$. If $T_1, \dots, T_\tau, Q^\perp$ are independent subspaces and there exists $\mb \Lambda$ satisfying $\mc P_{T_i} \mb \Lambda = \lambda_i \mb S_i$ and $\| \mc P_{T_i^\perp} \mb \Lambda \|_{(i)}^* < \lambda_i$ for each $i$, and $\PQp \mb \Lambda = \mb 0$,
then $\mb x_\star$ is the unique optimal solution to \eqref{eqn:general-c}.
\end{lemma}

Notice that this condition implies that $\mb \Lambda$ lies in the subdifferential of $\lambda_i \| \cdot \|_{(i)}$ for each $i$. The proof of Lemma \ref{lem:optimality} follows a familiar form, and is given in Appendix \ref{app:gen-duality}. Notice that if we take $Q = \Re^{m \times n}$ in Lemma \ref{lem:optimality}, we obtain a sufficient optimality condition for the original decomposition problem \eqref{eqn:general}. The condition given by Lemma \ref{lem:optimality} is not so convenient to directly work with, because it demands that $\mb \Lambda$ exactly satisfies a set of equality constraints $\PTi \mb \Lambda = \lambda_i \mb S_i$. One very useful device, due to Gross \cite{Gross2009-pp}, is to trade off between the equality constraints and the dual norm inequality constraints $\| \PTip \mb \Lambda \|_{(i)}^* < \lambda_i$, tightening the latter while loosening the former. The following definition gives this idea a name:

\begin{definition} \label{def:dec-cert} We call $\mb \Lambda$ an {\bf \em $(\alpha,\beta)$-inexact certificate} for a putative solution $(\mb X_{1,\star}, \dots, \mb X_{\tau,\star})$ to \eqref{eqn:general} with parameters $(\lambda_1, \dots, \lambda_\tau)$ if for each $i$, $\|\PTi \mb \Lambda - \lambda_i \mb S_i\|_F \;\le\; \alpha$, and $\| \PTip \mb \Lambda \|_{(i)}^* < \lambda_i \beta$.
\end{definition}

Comparing to the optimality condition in Lemma \ref{lem:optimality}, we can see that this definition is most meaningful when $\alpha$ is small, and $\beta \le 1$. Indeed, a number of simple and powerful analyses of problems such as matrix completion and robust low-rank matrix recovery proceed by constructing an inexact certificate for which $\alpha$ is polynomial in $m^{-1}$, and $\beta$ is a moderate constant, say, $1/2$. 

Definition \ref{def:dec-cert} pertains to the decomposition problem \eqref{eqn:general}, and does not involve the measurement operator $Q$ in any way. Adding one additional constraint, $\PQp \mb \Lambda = \mb 0$, we obtain an inexact certificate for the compressive decomposition problem \eqref{eqn:general-c}:

\begin{definition} \label{def:comp-dec-cert}  We call $\mb \Lambda$ an {\bf \em $(\alpha,\beta)$-inexact certificate} for a putative solution $(\mb X_{1,\star}, \dots, \mb X_{\tau,\star})$ to \eqref{eqn:general-c} with parameters $(\lambda_1, \dots, \lambda_\tau)$ if
\begin{itemize}
\item[(i)] $\mb \Lambda$ is an $(\alpha,\beta)$ inexact certificate for \eqref{eqn:general}, and
\item[(ii)] $\PQp \mb \Lambda = \mb 0$. 
\end{itemize}
\end{definition}

As we will see, an inexact certificate is easier to produce than the ``exact'' $\mb \Lambda$ demanded in the optimality condition Lemma \ref{lem:optimality}. Is it still sufficient to certify optimality? The following lemma shows the answer is {\em yes}, provided $\alpha$ and $\beta$ are small enough:

\begin{lemma}\label{lem:sc-relaxed} Consider a feasible solution $\mb x_\star = (\mb X_{1,\star}, \dots, \mb X_{\tau,\star})$ to the optimization problem \eqref{eqn:general-c}. Suppose that each of the norms $\| \cdot \|_{(i)}$ is decomposable at $\mb X_{i,\star}$, and that each of the $\| \,\cdot\, \|_{(i)}$ majorizes the Frobenius norm. Then if $T_1, \dots, T_\tau,Q^\perp$ are independent subspaces with 
\begin{equation}
\| \mc P_{T_i} \mc P_{T_j} \| \;<\; \frac{1}{\tau-1} \quad \forall \, i \ne j, 
\end{equation}
and there exists an $(\alpha,\beta)$-inexact certificate $\Lh$, with
\begin{equation}
\beta + \frac{\alpha \sqrt{\tau}}{(1- \| \mc P_{Q^\perp} \mc P_{T_1 + \dots + T_\tau} \|^2)\sqrt{ 1 - (\tau-1) \max_{ij} \| \mc P_{T_i} \mc P_{T_j} \|}} \times \frac{1}{\min_l \lambda_l} < 1,
\end{equation}
then $\mb x_\star$ is the unique optimal solution. 
\end{lemma}

We prove this lemma in Section \ref{sec:inex-ex-upgrade}, using a least squares perturbation argument. The additional technical condition that $\|\cdot\|_{(i)}$ majorizes the Frobenius norm (i.e., for all $\mb X$, $\| \mb X \|_{(i)} \ge \| \mb X \|_F$) is immediately satisfied by sparsity inducing norms such as the nuclear and $\ell^1$ norms. In any case, it can always be ensured by rescaling. 

\begin{remark}The denominator in the condition of Lemma \ref{lem:sc-relaxed} depends on our knowledge of the relative orientation of the subspaces $T_1, \dots, T_\tau$ and $Q$. We have stated the lemma in a way that assumes bounds on the angles of each pair $(T_i,T_j)$ and between $T_1 + \dots + T_\tau$ and $Q$, but demands no additional knowledge. A tighter accounting is possible if more is known about the configuration of $(T_1,\dots,T_\tau,Q)$. 
\end{remark}

Thus, to show that $\mb X_1, \dots, \mb X_\tau$ solve the compressive decomposition problem \eqref{eqn:general-c}, we just have produce an inexact certificate $\mb \Lambda$ following the specification of Definition \ref{def:comp-dec-cert} with $(\alpha,\beta)$ sufficiently small. This is fortuitous, since many existing analyses of the original decomposition problem \eqref{eqn:general} already give certificates for that problem. For example, for Principal Component Pursuit, we can leverage existing constructions in \cite{Candes2011-JACM}. To prove that the desired solution remains optimal even when we only see a few measurements $Q$, we will show that a certificate for \eqref{eqn:general} can be ``upgraded'' to a certificate for \eqref{eqn:general-c}, with very high probability in the choice of random $Q$, and only a small loss in the parameters $(\alpha,\beta)$. 

Of course, intuitively speaking, this should only be possible if the number of measurements is sufficient: if the number of measurements in $Q$ is smaller than the number of degrees of freedom in $\mb x_\star$, then reconstruction from the compressive measurements $\PQ \mb M$ should not be possible. Interestingly, however, we will see that the number of measurements does not need to be too much larger than the number of degrees of freedom in $\mb x_\star$: oversampling by $O(\log^2 m)$ will suffice. We have been a bit vague about what we mean by the number of degrees of freedom in the signal. To be precise, our theorem will refer to the quantity $\dim{T_1 + \dots + T_\tau}$. Indeed, for the $\ell^1$ norm, $\dim{T_i}$ is the number of nonzero entries in the solution $\mb X_i$. For the nuclear norm, one can check that $\dim{T_i}$ is the number of degrees of freedom in specifying a matrix whose rank is equal to that of $\mb X_i$. 

Our main theorem states that with very high probability it is possible to ``upgrade'' a certificate for the decomposition problem \eqref{eqn:general} to one for the compressive decomposition problem \eqref{eqn:general-c}, with only small loss in parameters $(\alpha,\beta)$. As it turns out, the loss in the dual norm $\| \cdot \|_{(i)}^*$ will be bounded by the expected dual norm of a standard Gaussian matrix. We will let $\nu_i$ denote this quantity: 
\begin{equation}
\nu_i \doteq \expect{\| \mb G \|_{(i)}^*}, \qquad \mb G \sim_{iid} \normal{0}{1}.
\end{equation}
We have the following theorem: 

\begin{theorem} [\bf Certificate Upgrade] \label{thm:upgrade} Consider the general decomposition problem \eqref{eqn:general}, and suppose that each of the norms $\| \cdot \|_{(i)}$ majorizes the Frobenius norm. Let $\mb x_\star = (\mb X_{1,\star}, \dots, \mb X_{\tau,\star})$ be feasible for \eqref{eqn:general}, and suppose there exists an $(\alpha,\beta)$-inexact certificate for $\mb x_\star$ for the decomposition problem \eqref{eqn:general} with parameters $(\lambda_i)$. 

Then if $Q \subset \Re^{m \times n}$ is a random subspace distributed according to the Haar measure, with
\begin{equation} \label{eqn:Q-perturb-meas}
\dim{Q} \;\ge\; C_Q \cdot \dim{T_1 + \dots + T_\tau} \cdot \log^2 \, m,
\end{equation}
there exists an $(\alpha',\beta')$-inexact certificate for $\mb x_\star$ for the {\em compressive} decomposition problem \eqref{eqn:general-c} with 
\begin{eqnarray}
\alpha' &\le& \alpha + m^{-3} \|\Wh \|_F, \\
\beta'  &\le& \beta  +  C_1 \max_i \frac{\nu_i + \sqrt{\log m}}{\lambda_i} \left(  \frac{\|\Wh\|_F^2 \log m}{\dim{Q}} \right)^{1/2}, \label{eqn:beta-perturb}
\end{eqnarray}
with probability at least $1-C_2 \cdot \tau \cdot m^{-9}$ in $Q$. Above, $C_Q$, $C_1$ and $C_2$ are positive numerical constants. 
\end{theorem}

\begin{remark} As will become clear in the proof, the degrees $m^{-3}$ and $m^{-9}$ above are arbitrary, and can be set to be any constants by appropriate choice of $C_Q, C_1, C_2$. 
\end{remark}

\begin{remark} [\bf Scaling in \eqref{eqn:beta-perturb}] A casual glance at \eqref{eqn:beta-perturb} may suggest that we can make $\beta'$ arbitrarily close to $\beta$, by setting $(\lambda_i)$ large. This is actually not the case: the initial certificate $\Lh$ must satisfy $\PTi \Lh \approx \lambda_i \mb S_i$. Scaling all of the $\lambda_i$ by the same amount will also scale $\Lh$, causing no effective change to the right hand side of \eqref{eqn:beta-perturb}. 

On the other hand, Theorem \ref{thm:upgrade} suggests an interesting practical role for the expected norms $\nu_i$ in choosing the {\em relative} values of $\lambda_i$. Namely, it suggests setting $\lambda_i \propto \nu_i$. This is consistent (within logarithmic factors) with suggestions in \cite{Candes2011-JACM}, and could suggest a principled way of combining many such structure-inducing norms as in \eqref{eqn:general-c}. 
\end{remark}

\newcommand{\Eerr}{\event_{\mathrm{err}}}
\newcommand{\Einfty}{\event_\infty}
\newcommand{\Eop}{\event_{\mathrm{op}}}
\newcommand{\Egood}{\event_{\mathrm{good}}}

\begin{proof}[\bf \em Proof of Theorem \ref{thm:upgrade}] Let 
\begin{equation}
S = T_1 + \dots + T_\tau + \mathrm{span}(\Wh),
\end{equation}
where $+$ denotes Minkowski summation. Then $S$ is a linear subspace of dimension at most $\dim{T_1 + \dots + T_\tau} + 1$ containing $\Wh$. Our goal is to generate a certificate $\cert{\star}$ that is close to $\Wh$ on $S$, and also satisfies 
\begin{equation}\label{eqn:Wstar-goal}
\PQp\cert{\star} = \mb 0.
\end{equation}
 Such a $\cert{\star}$ would inherit the good properties of $\Wh$ on $T_1 + \dots + T_\tau \subseteq S$, and also satisfy the additional equality constraint \eqref{eqn:Wstar-goal} -- in effect, certifying that the measurements are sufficient. 

To this end, we will set
\begin{equation}
\cert{0} \;=\; \mb 0,
\end{equation}
We will generate inductively a sequence $(\cert{j})_{j = 1, \dots, k}$ for appropriate $k$, such that with high probability $\cert{\star} = \cert{k}$ is the desired certificate. The initial guess $\cert{0}$ obviously satisfies \eqref{eqn:Wstar-goal}, but could be very far from $\Wh$ on $S$. Define the error at step $j$ to be 
\begin{equation}
\err{j} \;=\; \PS [ \cert{j} ] - \Wh \;\in\; S. 
\end{equation}
We will generate a sequence of corrections, each of which lies in $Q$, that drive $\err{j}$ toward zero. 

By orthogonal invariance, $Q$ is equal in distribution to the linear span of $$\mb H_1, \dots, \mb H_{\dim{Q}},$$
where $\mb H_j$ are independent iid $\normal{0}{1/mn}$ random matrices. Choose from $\set{ 1, \dots, \dim{Q} }$, 
\begin{equation} 
k = \ceiling{ 3 \log_2 m }
\end{equation}
disjoint subsets $I_1, \dots, I_k$ of size 
\begin{equation}
\gamma \;=\; \floor{ \frac{\dim{Q}}{k} }.
\end{equation}
Our choice of constant ensures that $2^{-k} \le m^{-3}$. We will require that 
\begin{equation} \label{eqn:gamma-bound}
\gamma \;\ge\; C_3 \cdot \dim{S}, 
\end{equation}
where $C_3$ is a numerical constant to be specified later. Since by assumption $$\dim{Q} \;\ge\; C_Q \cdot \dim{T_1 + \dots + T_\tau} \cdot \log^2(m),$$ once $C_3$ is chosen, we can ensure that $C_Q$ is large enough that \eqref{eqn:gamma-bound} holds. 

Let $\mc A_j : \Re^{m \times n} \to \Re^{m \times n}$ denote the semidefinite operator that acts via
\begin{equation}
\mc A_j[ \cdot ] = \sum_{i \in I_j} \mb H_i \< \mb H_i, \cdot \>.
\end{equation}
Notice that $\expect{ \mc A_j } = \frac{\gamma}{mn} \mc I$. For $j = 1, \dots, k$, let 
\begin{eqnarray}
\cert{j} &=& \cert{j-1} - \partop{j} \err{j-1} \nonumber \\
&=& - \sum_{i=1}^j \partop{i} \err{i-1}. \label{eqn:W-sum}
\end{eqnarray}
Then we have 
\begin{eqnarray}
\err{j} 
	&=& \PS[ \cert{j} ] - \Wh \nonumber \\ 
	&=& \PS[ \cert{j-1} ] - \Wh - \PS \partop{j} \err{j-1} \nonumber \\
	&=& \PS \contract{j} \PS \err{j-1}. \label{eqn:E-contract}
\end{eqnarray}
In paragraph (i) below, we will use this expression to control the Frobenius norm of $\err{k}$. We may further write 
\begin{eqnarray}
\cert{k} 
	&=& \PS [ \cert{k} ] + \PSp [ \cert{k} ], \nonumber \\
	&=& \Wh + \err{k} - \sum_{j=1}^k \PSp \partop{j} \PS \err{j-1}, \label{eqn:cert-Sp}
\end{eqnarray}
where in \eqref{eqn:cert-Sp} we have used that $\err{j} \in S$ for all $j$. In paragraphs (ii)-(iii) below, we will use this final expression to control the dual norms of $\cert{k}$.

\paragraph{(i) Driving $\mb E$ to zero.}

Ensuring that $C_3$ in \eqref{eqn:gamma-bound} is sufficiently large that the hypotheses of Lemma \ref{lem:op-appx} are verified, we have that with probability at least $1 - C_4 \exp(- c_1 \gamma )$, 
\begin{equation}
\norm{\PS \contract{j} \PS}{} \;=\; \norm{\PS \partop{j} \PS - \PS}{} \;\le\; \frac{1}{2}. 
\end{equation}
Hence, by \eqref{eqn:E-contract}, we have $\norm{\err{j}}{F} \le \tfrac{1}{2} \norm{\err{j-1}}{F}$ for each $j$, on the complement of a bad event $\Eerr$ of probability at most $C_4 k \exp(- c_1 \gamma )$. On $\Eerr^c$, 
$\norm{\err{j}}{F} \le 2^{-j} \norm{\err{0}}{F}$
for each $j$, giving 
\begin{equation}
	\norm{\err{k}}{F} 
		\le 2^{-k} \|\Wh\|_F 
		\quad \text{and} \quad 
	\sum_{j=0}^k \norm{\err{j}}{F} 
		\le 2 \| \Wh\|_F. 
\end{equation}

\paragraph{(ii) Analysis of $\alpha'$.} From the definition, we may set $\alpha' = \max_i \| \PTi \cert{k} - \lambda_i \mb S_i \|_F$. On $\Eerr^c$,
\begin{eqnarray}
	\norm{\PTi \cert{k} - \lambda_i \mb S_i }{F} &=& \|\PTi[ \Wh + \err{k} ] - \lambda_i \mb S_i\|_F,  \nonumber \\
		&\le& \| \PTi \Wh - \lambda_i \mb S_i \|_F  + \norm{ \err{k} }{F},           \nonumber \\
		&\le& \| \PTi \Wh - \lambda_i \mb S_i \|_F  + 2^{-k} \|\Wh\|_F,                     \nonumber \\
		&\le& \| \PTi \Wh - \lambda_i \mb S_i \|_F  + m^{-3} \|\Wh\|_F.
\end{eqnarray}
Since the first term is bounded by $\alpha$, the claim in \eqref{eqn:first} is established.

\paragraph{(iii) Analysis of $\beta'$.} Similarly, for $\beta'$, we can take 
\begin{equation}
\beta' = \max_{i = 1, \dots, \tau} \; \lambda_i^{-1} \| \cert{k} \|_{(i)}^* 
\end{equation}
From \eqref{eqn:cert-Sp} and the triangle inequality, we have 
\begin{eqnarray}
	\norm{\cert{k}}{(i)}^* 
		&\le& \| \Wh\|_{(i)}^* + \norm{\err{k}}{F} + \sum_{j=1}^k \norm{\PSp \partop{j} \PS \err{j-1} }{(i)}^*,
\end{eqnarray}
where we used the fact that whenever the primal norm majorizes the Frobenius norm, its dual minorizes the Frobenius norm. Applying Lemma \ref{lem:skew-infty}, this is bounded by 
\begin{eqnarray}
	\norm{\cert{k}}{(i)}^* 
		&\le& \lambda_i \beta + 2^{-k} \norm{\err{0}}{F} + 10 \frac{\nu_i + \sqrt{\log m}}{\sqrt{\gamma}} \sum_{j=1}^k \norm{\err{j-1}}{F}, \nonumber \\
		&\le& \lambda_i \beta + \left( 2^{-k} + 20 \frac{\nu_i + \sqrt{\log m}}{\sqrt{\gamma}} \right) \norm{\err{0}}{F}, \nonumber \\
		&\le& \lambda_i \beta + 21 \frac{\nu_i + \sqrt{\log m}}{\sqrt{\gamma}} \|\Wh\|_F
\end{eqnarray}
on the complement of an event $\Einfty$ of probability at most $2 k m^{-10} + k \expfrac{\gamma}{2} + \prob{\Eerr}$. In the final line, we have used that $2^{-k} \le m^{-3}$ and $\gamma \le m^2$. Since $\gamma \;\ge\; \frac{c \cdot \dim{Q}}{\log m}$, for some numerical constant $C_5$,
\begin{equation}
\lambda_i^{-1} \norm{\cert{k}}{(i)}^* \;\le\; \beta + C_5 \frac{\nu_i + \sqrt{\log m}}{\lambda_i} \left( \frac{ \| \Wh \|_F^2 \,\log m  }{\dim{Q}} \right)^{1/2}. \label{eqn:infty-final}
\end{equation}
Taking a union bound over $i = 1, \dots, \tau$, we have that the desired bounds on $\alpha'$ and $\beta'$ hold simultaneously on the complement of an event of probability
\begin{equation}
2k \tau m^{-10} + k \tau \exp\left(-\frac{\gamma}{2}\right) + C_4 (1+\tau)  k \exp\left(-c_1 \gamma\right). 
\end{equation}
Since $\gamma = \Omega(C_Q \log m)$, provided $C_Q$ is large enough, the exponential terms will also be bounded by $m^{-10}$. Using that $k = O(\log m)$, we obtain the promised bound on the probability of failure. 
\end{proof}

%

\section{Key Probabilistic Lemmas}

This section introduces two probabilistic lemmas used in the analysis of the golfing scheme. The first lemma shows that $\PS \partop{j} \PS \approx \PS$. Its (routine) proof is delayed to the appendix. The second lemma is crucial for controlling the dual norms of $\PSp \partop{j} \PS$, and is proved in this section. 

\begin{lemma} \label{lem:op-appx} There exist numerical constants $C_1,C_2,c > 0$ such that the following holds. Let $S \subseteq \Re^{m \times n}$ be a fixed linear subspace, and let $\mc A = \sum_{j=1}^\gamma \mb H_j \< \mb H_j, \cdot \>$, where $(\mb H_j)$ is a sequence of independent iid $\mc N(0,1/mn)$ random matrices, and let $R =\mathrm{range}(\mc A) \subseteq \Re^{m \times n}$. Then if 
\begin{equation}
\gamma \;\ge\; C_1 \cdot \dim{S},
\end{equation}
with probability at least $1 - C_2 \exp\left( -c \gamma \right)$,
\begin{equation} \label{eqn:first}
\norm{ \PS \partop{} \PS - \PS }{} \;\le\; \frac{1}{2},
\end{equation}
and 
\begin{equation} \label{eqn:second}
\norm{ \PS \mc P_R \PS - \frac{\gamma}{mn} \PS }{} \;\le\; \frac{1}{16} \frac{\gamma}{mn}.
\end{equation}
\end{lemma}

The proof of this result follows a familiar covering argument. For completeness, we give this proof in Appendix \ref{app:op-appx}. We next state and prove the key probabilistic lemma for analyzing the ``upgrade'' procedure introduced in the proof of Theorem \ref{thm:upgrade}. This lemma allows us to control the dual norm of the constructed certificate. The result is as follows:

\begin{lemma} \label{lem:skew-infty} Let $S$ be any fixed subspace of $\Re^{m \times n}$ ($m \ge n$), $\mb M$ any fixed matrix. Let $\mc A = \sum_{l=1}^\gamma \mb H_l \< \mb H_l, \cdot \>$ be a random semidefinite operator constructed from a sequence of independent iid $\mc N(0,1/mn)$ matrices $\mb H_1, \dots, \mb H_\gamma$. Let $\| \cdot \|$ be any norm that majorizes the Frobenius norm, and let $\| \cdot \|^*$ be its dual norm. Set $\nu = \expect{\, \| \mb G \|^* }$, with $\mb G$ iid $\normal{0}{1}$. Then we have
\begin{equation}
\norm{\PSp \partop{} \PS \mb M}{}^* \;\le\; 10 \norm{\PS \mb M}{F} \frac{\nu + \sqrt{\log m}}{\sqrt{\gamma}},
\end{equation} 
with probability at least $1 - m^{-10} - \expfrac{\gamma}{2}$.
\end{lemma}

\begin{proof}

Let $\Gamma \subseteq \{ \mb N \mid \| \mb N \|_F \le 1 \}$ denote the unit ball for $\| \cdot\|$. Then 
 $$\| \PSp \mc A \PS \mb M\|_{}^* = \sup_{\mb N \in \Gamma} \innerprod{ \mb N }{ \PSp \mc A \PS \mb M}.$$ Inner products of this form are particularly easy to control because they involve projections of $\mc A$ onto orthogonal subspaces. Since $S$ and $S^\perp$ are orthogonal and $\mb H_l$ is iid Gaussian, $\PS \mb H_l$ and $\PSp \mb H_l$ are probabilistically independent. So, letting $\mb H'_1, \dots, \mb H'_\gamma$ denote an independent copy of $\mb H_1, \dots, \mb H_\gamma$, we have
\begin{eqnarray}
\PSp \mc A \PS [\,\cdot\,] &=& \PSp \sum_l \mb H_l \innerprod{ \mb H_l }{ \PS[\,\cdot\,] } \nonumber \\
&=& \sum_l \PSp[\mb H_l] \innerprod{ \PS \mb H_l }{\,\cdot\, } \nonumber \\
&\eqdist& \sum_l \PSp[\mb H_l] \innerprod{\PS \mb H_l'}{\, \cdot \,} \quad\doteq\quad \mc D, 
\end{eqnarray}
where $\eqdist$ denotes equality in distribution. Hence, we have 
$$\|\PSp \mc A \PS \mb M\|^* \eqdist \|\mc D \mb M\|^*.$$ 
We find the second term more convenient to analyze. Conditioned on $\mb H'_1, \dots, \mb H'_\gamma$, 
\begin{eqnarray}
\xi_{\mb N} \;\doteq\; \innerprod{\mb N}{ \mc D \mb M } \;= \sum_l \innerprod{ \PS \mb H_l' }{  \mb M  } \innerprod{ \PSp \mb N }{ \mb H_l  } 
\end{eqnarray}
is zero-mean Gaussian. We have $\| \mc D \mb M \|^* = \sup_{\mb N \in \Gamma} \xi_{\mb N}$. A quick calculation shows that for any $\mb N$ and $\mb N'$, 
\begin{eqnarray}
\condexp{ (\xi_{\mb N} - \xi_{\mb N'})^2 }{\mb H_1', \dots, \mb H_\gamma'} &=& \frac{\norm{\PSp (\mb N - \mb N')}{F}^2}{mn} \sum_{l=1}^\gamma \innerprod{\PS \mb H_l'}{\mb M}^2 \nonumber \\
&\le& \frac{\|\mb N - \mb N'\|_F^2}{mn}\sum_{l=1}^\gamma \innerprod{\PS \mb H_l'}{\mb M}^2 \;=\; \frac{\Xi^2 \| \mb N - \mb N'\|_F^2}{mn}, 
\end{eqnarray}
where we have let 
\begin{equation}
\Xi \;=\; \skewstdev
\end{equation}
Consider a second zero-mean Gaussian process $(\zeta_{\mb N})_{\mb N \in \Gamma}$, defined by letting $\mb G$ be an iid $\normal{0}{1/mn}$ matrix, and setting $\zeta_{\mb N} \;=\; \< \mb N, \mb G \>$. From the definition of $\nu$, 
\begin{equation}
\expect{ \sup_{\mb N \in \Gamma} \zeta_{\mb N} } = \frac{\nu}{\sqrt{mn}}.
\end{equation}
Another calculation shows that 
\begin{equation}
\expect{ (\zeta_{\mb N} - \zeta_{\mb N'})^2 } = \frac{\| \mb N - \mb N' \|_F^2}{mn}.
\end{equation}
By Slepian's inequality (e.g., \cite{Vershynin2011} Lemma 5.33, \cite{Ledoux-Talagrand} Chapter 3), we have 
\begin{equation}
\expect{ \sup_{\mb N} \, \xi_{\mb N} \mid \mb H'_1, \dots, \mb H'_\gamma } \;\le \;\Xi \cdot \expect{ \sup_{\mb N} \, \zeta_{\mb N} } \;=\; \frac{\nu \, \Xi}{\sqrt{mn}}. 
\end{equation}
Moreover, for any fixed values of $\mb H_1', \dots, \mb H_l'$, and any $\mb N \in \Gamma$,  $\xi_{\mb N}$ is a $\Xi$-Lipschitz function of the iid Gaussian sequence $(\mb H_1, \dots, \mb H_l)$. Hence, the supremum $\| \mc D \mb M \|^*$ is also $\Xi$-Lipschitz. By Lipschitz concentration (\cite{Ledoux} Proposition 2.18), 
\begin{equation}
\condprob{ \| \mc D \mb M \|^* > \condexp{ \| \mc D \mb M \|^* }{(\mb H'_l)} + \frac{t \, \Xi}{\sqrt{mn}} }{(\mb H'_l)} \le \expfrac{t^2}{2}.
\end{equation}
Combining with our previous estimates, we have 
\begin{equation} 
\condprob{ \| \mc D \mb M \|^* > \frac{ \Xi \, ( \nu + t ) }{\sqrt{mn}} }{(\mb H'_l)} \le \expfrac{t^2}{2}.
\end{equation}
Since this estimate holds for any value of $(\mb H'_l)$, it holds unconditionally:
\begin{equation} \label{eqn:DM-bound}
\prob{ \| \mc D \mb M \|^* > \frac{ \Xi \, ( \nu + t ) }{\sqrt{mn}} } \le \expfrac{t^2}{2}.
\end{equation}
Moreover, it is easy to notice that 
 $\Xi$ is itself a $\norm{ \PS \mb M  }{F}$-Lipschitz function of the iid Gaussian sequence $(\mb H'_1, \dots, \mb H'_\gamma)$,
 with  
\begin{equation}
\expect{\Xi} \le \left(\expect{\Xi^2}\right)^{1/2} \;=\; \norm{ \PS \mb M }{F} \sqrt{\frac{\gamma}{ mn } }.
\end{equation}
From Lipschitz concentration,
\begin{equation}
\prob{\Xi > \expect{\Xi} + s \norm{ \PS \mb M }{F} } \;\le\; \exp\left( -\frac{s^2 mn}{ 2} \right), 
\end{equation}
and so 
\begin{equation} \label{eqn:skew-st-dev}
\prob{ \Xi \;>\; 2 \norm{ \PS \mb M }{F} \sqrt{\frac{\gamma}{mn}} 
} \le \expfrac{\gamma}{2}.
\end{equation}
Combining this estimate with \eqref{eqn:DM-bound}, setting $t = \sqrt{20 \log m}$, and rescaling by $\tfrac{mn}{\gamma}$, we obtain the result.
\end{proof}

\section{Proof of Lemma \ref{lem:sc-relaxed}: Upgrade to Exact Certificate} \label{sec:inex-ex-upgrade}

In this section, we prove Lemma \ref{lem:sc-relaxed}, which shows how an inexact dual certificate can be upgraded to an exact certificate of optimality. The result will be a consequence of the following general lemma on systems of equations. 

\begin{lemma} Let $T_1, \dots , T_k$ be independent subspaces of $\Re^{m \times n}$, and $\mb S_1 \in T_1, \dots, \mb S_k \in T_k$. Then the system of equations 
\begin{equation} \label{eqn:mtx-system}
\mc P_{T_i} \mb X = \mb S_i, \; i = 1, \dots, k,
\end{equation}
has a solution $\mb X \in T_1 + \dots + T_k$ satisfying 
\begin{equation}
\|\mb X \|_F \;\le\; \sqrt{\frac{ \sum_i \| \mb S_i \|_F^2 }{ 1 - (k-1) \max_{i \ne j} \| \mc P_{T_i} \mc P_{T_j} \|}}. 
\end{equation}
\end{lemma}
\begin{proof} Let $\mathrm{vec} : \Re^{m\times n} \to \Re^{mn}$ denote the operator that vectorizes a matrix by stacking its columns. For each $i$, let $\mb U_i \in \Re^{mn \times \dim{T_i}}$ denote a matrix whose columns form an orthonormal basis for $\vec{T_i}$. The system of equations is equivalent to 
\begin{equation}
\left[ \begin{array}{c} \mb U_1^* \\ \vdots \\ \mb U_k^* \end{array} \right] \mb x = \left[ \begin{array}{c} \mb U_1^* \cdot \vec{\mb S_1} \\ \vdots \\ \mb U_k^* \cdot \vec{\mb S_k} \end{array} \right],
\end{equation}
where $\mb x = \vec{\mb X}$. Let $\mb U^*$ denote the matrix on the left hand side, and $\mb s$ denote the vector on the right hand side. Then $\| \mb s \|_2^2 = \sum_{i=1}^k \| \mb S_i \|_F^2$. The system of equations has a solution $\mb x \in \mathrm{range}(\mb U) = \vec{T_1 + \dots + T_k}$ with $\ell^2$ norm at most $\| \mb s \|_2 / \sigma_{\min}(\mb U)$, and hence \eqref{eqn:mtx-system} has a solution whose Frobenius norm is bounded by the same quantity. Write
\begin{equation}
\mb U^* \mb U = \left[ \begin{array}{cccc} \mb I & \mb U_1^* \mb U_2 & \dots & \mb U_1^* \mb U_k \\ \mb U_2^* \mb U_1 & \mb I & \dots & \mb U_2^* \mb U_k \\ 
\vdots & \vdots & \ddots & \vdots \\
\mb U_k^* \mb U_1 & \mb U_k^* \mb U_2 & \dots & \mb I \end{array} \right]. 
\end{equation}
Let $\lambda$ be any eigenvalue of $\mb U^* \mb U$, with corresponding eigenvector $\mb x = (\mb x_1^*, \;\mb x_2^*,\; \dots, \mb x_k^*)^*$. Let $p = \arg \max_j \| \mb x_j\|_2$. Then, looking at just the $p$-th block of the equation $\lambda \mb x = \mb U \mb U^* \mb x$, we have
\begin{eqnarray}
|\lambda - 1| \| \mb x_p \|_2 &=& \Bigl\| \sum_{j \ne p} \mb U_p^* \mb U_j \mb x_j \Bigr\|_2 \nonumber \\ &\le& \sum_{j \ne p} \| \mb U_p^* \mb U_j \| \| \mb x_j \|_2 \nonumber \\ &\le& \| \mb x_p \|_2 \times (k-1) \max_{i \ne j} \| \mb U_i^* \mb U_j \|.
\end{eqnarray}
Since $\| \mb U_i^* \mb U_j \| = \| \mc P_{T_i} \mc P_{T_j} \|$, we conclude that $\lambda_{\min}(\mb U^* \mb U) \ge 1 - (k-1) \max_{i \ne j} \| \mc P_{T_i} \mc P_{T_j} \|$, and hence $\sigma_{\min}(\mb U)$ is at least as large as the square root of this quantity. This establishes the result. 
\end{proof}

\begin{proof}[\bf \em Proof of Lemma \ref{lem:sc-relaxed}]
The assumption implies that $T_1, \dots, T_\tau$ are independent subspaces, and so the system of equations 
\begin{equation}
\mc P_{T_i} \mb \Delta = \lambda_i \mb S_i - \mc P_{T_i} \Lh, \quad  i = 1, \dots, \tau
\end{equation}
is feasible, and has a solution $\mb \Delta_0 \in T_1 + \dots + T_\tau$ of Frobenius norm at most 
\begin{equation}
\| \mb \Delta_0 \|_F \;\le\; \sqrt{\frac{\sum_i \|\lambda_i \mb S_i - \mc P_{T_i} \Lh \|_F^2 }{1 - (\tau-1) \max_{i \ne j} \| \mc P_{T_i} \mc P_{T_j} \|}} \; \le \; \sqrt{\frac{\alpha^2 \tau}{1 - (\tau-1) \max_{i \ne j} \| \mc P_{T_i} \mc P_{T_j} \|} }.
\end{equation}
Moreover, since $T_1 + \dots + T_\tau$ and $Q^\perp$ are independent, the system of equations 
\begin{equation}
\mc P_{T_1 + \dots + T_\tau} \mb \Delta = \mb \Delta_0, \; \PQp \mb \Delta = \mb 0
\end{equation}
is feasible (indeed, underdetermined). We consider a solution $\mb \Delta_\star$ of minimum Frobenius norm. Under the stated hypotheses, this solution is given by the Neumann series 
\begin{equation}
\mb \Delta_\star = \PQ \sum_{i = 0}^{\infty} ( \mc P_{T_1 + \dots + T_\tau} \PQp \mc P_{T_1 + \dots + T_\tau} )^i \mb \Delta_0,
\end{equation}
whose norm is bounded as 
\begin{equation}
\| \mb \Delta_\star \|_F \;\le\; \frac{\| \mb \Delta_0 \|_F}{1 - \| \mc P_{T_1 + \dots + T_\tau} \PQp \|^2}.
\end{equation}
Set $\mb \Lambda = \Lh + \mb \Delta_\star$, and observe that by construction, for each $i$, $\PTi \mb \Lambda = \lambda_i \mb S_i$. For each $i$, we have
\begin{equation}
\| \PTip \cert{} \|_{(i)}^* \quad\le\quad  \| \PTip \Lh \|_{(i)}^* \,+\, \| \PTip \mb \Delta_\star \|_{(i)}^*. 
\end{equation}
Because $\| \cdot \|_{(i)}$ majorizes the Frobenius norm, its dual minorizes the Frobenius norm, and so we have 
\begin{equation}
\lambda_i^{-1} \| \PTip \cert{} \|_{(i)}^* \quad\le\quad \lambda_i^{-1} \| \PTip \Lh \|_{(i)}^* \,+\, \lambda_{i}^{-1} \| \PTip \mb \Delta_\star \|_F.
\end{equation}
Under the stated hypotheses, this quantity is strictly smaller than one, and so $\cert{}$ satisfies the conditions of Lemma \ref{lem:optimality}.
\end{proof}

\section{Proof of Theorem \ref{thm:main}: Compressive PCP Recovery}

In this section, we prove Theorem \ref{thm:main}, using the general upgrade provided by Theorem \ref{thm:upgrade}. In the language of Section \ref{sec:upgrade}, we have $\| \cdot \|_{(1)} = \| \cdot \|_*$, $\| \cdot \|_{(2)} = \| \cdot \|_1$. Both of these norms majorize the Frobenius norm. For PCP, we take $\lambda_1 = 1$, $\lambda_2 = 1/ \sqrt{m}$.

Let $\mb L_0,\mb S_0$ denote the target pair, and $r$ the rank of $\mb L_0$. If we let $\mb L_0 = \mb U \mb S \mb V^*$ denote the rank-reduced singular value decomposition of $\mb L_0$, and $T$ denote the subspace
\begin{equation}
T \doteq \set{ \mb U \mb X^* + \mb Y \mb V^* \mid \mb X \in \Re^{n \times r}, \; \mb Y \in \Re^{m \times r} },
\end{equation}
then the subdifferential of the nuclear norm at $\mb L_0$ is 
\begin{equation}
\partial \| \cdot \|_*(\mb L_0) \;=\; \set{ \mb \Lambda \mid \PT \mb \Lambda = \mb U \mb V^*, \; \|\PTp \mb \Lambda \| \le 1 }.
\end{equation}
It is easy to check that $\PTp : \mb M \mapsto (\mb I - \mb U \mb U^*) \mb M (\mb I - \mb V \mb V^*)$ is nonexpansive with respect to the operator norm $\| \cdot \|$, and so $\| \cdot \|_*$ indeed satisfies our criteria for a decomposable norm. Similarly, if we let $\Omega = \mathrm{supp}(\mb S_0)$ and $\mb \Sigma = \mathrm{sign}(\mb S_0)$, then 
\begin{equation}
\partial \| \cdot \|_1(\mb S_0) \;=\; \set{ \mb \Lambda \mid \PO \mb \Lambda = \mb \Sigma, \; \| \POc \mb \Lambda \|_\infty \le 1 }.
\end{equation}
Again, $\POc$ does not increase the $\ell^\infty$ norm, and $\| \cdot \|_1$ is also decomposable. In the language of Theorem \ref{thm:upgrade}, we have $T_1 = T$, $\mb S_1 = \mb U \mb V^*$, $T_2 = \Omega$, $\mb S_2 = \mb \Sigma$. 
For the PCP problem, an $(\alpha,\beta)$-inexact certificate is therefore a matrix $\Lpcp$ satisfying 
\begin{eqnarray*}
\| \PT \Lpcp - \mb U\mb V^* \|_F &\le& \alpha, \\
\| \PO \Lpcp - \lambda \mb \Sigma \|_F &\le& \alpha, \\
\| \PTp \Lpcp \| &\le& \beta, \\
\| \POc \Lpcp \|_\infty &\le& \beta \lambda.
\end{eqnarray*}
Such a certificate was constructed in \cite{Candes2011-JACM},\footnote{In the notation of \cite{Candes2011-JACM}, the certificate constructed there is $\Lpcp = \mb U \mb V^* + \mb W^L + \mb W^S$.} under the hypotheses of Theorem \ref{thm:main}. More precisely, we have the following:

\begin{proposition}[\bf Dual Certification for PCP \cite{Candes2011-JACM}] \label{thm:clmw} Under the conditions of Theorem \ref{thm:main}, on an event of probability at least $1 - C m^{-10}$ the following hold:
\begin{equation} \label{eqn:popt-bound}
\text{(i)} \qquad \norm{\PO\PT}{} \;\le\; 1/2,
\end{equation}
and (ii) there exists a $(m^{-2},1/4)$-inexact PCP certificate $\Lpcp$ for $(\mb L_0,\mb S_0)$, which satisfies
\begin{equation} \label{eqn:Wh-bound}
\|\Lpcp\|_F \;\le\; 4 \sqrt{\mathrm{rank}(\mb L_0)} + (4/3) \lambda \sqrt{\|\mb S_0\|_0}.
\end{equation}
Above, $C$ is numerical.
\end{proposition}

The careful reader may notice that the relaxation parameters $(\alpha,\beta)$ in Proposition \ref{thm:clmw} are stricter than those provided by \cite{Candes2011-JACM}, which gives $\alpha = 1/ 4 \sqrt{m}$, $\beta = 1/2$. In fact, by modifying the constants in the construction of \cite{Candes2011-JACM}, we can achieve $\beta$ smaller than any desired constant, and $\alpha$ smaller than any polynomial in $m^{-1}$, at the expense of slightly more stringent (but qualitatively equivalent) demands on $(\mb L_0,\mb S_0)$. The bound \eqref{eqn:Wh-bound} is implied by the probabilistic lemmas in \cite{Candes2011-JACM}, but requires a bit of manipulation to obtain. Below, we will first prove Theorem \ref{thm:main}, and then sketch a proof of the modifications to \cite{Candes2011-JACM} needed to obtain the supporting result Proposition \ref{thm:clmw}. 

\begin{proof}[\bf \em Proof of Theorem \ref{thm:main}] 
From Lemma \ref{lem:sc-relaxed}, to show that $(\mb L_0,\mb S_0)$ is the unique optimal solution to the compressive PCP problem, it is enough to show that 
\begin{itemize}
\item[{\bf (I)}] $\| \PT \PO \| \le 1/2$. 

\item[{\bf (II)}] There exists an $(\alpha',1/2)$-inexact CPCP certificate $\cert{\mathrm{CPCP}}$ with 
$\alpha' \;<\; \frac{1 - \| \PQp\PTO \|^2}{4 \sqrt{m}}$. 
\end{itemize}
We accomplish this in three parts. In paragraph (i) below, we apply Lemma \ref{lem:op-appx} to lower bound $1 - \| \PQp \PTO \|^2$. In paragraph (ii), we use Proposition \ref{thm:clmw} to show (I) and the existence of an inexact PCP certificate $\cert{\mathrm{PCP}}$. In paragraph (iii) we use Theorem \ref{thm:upgrade} to upgrade this to an inexact CPCP certificate $\cert{\mathrm{CPCP}}$ that satisfies property (II). Paragraph (iv) completes the proof by showing that the probability of failure is appropriately small. 

\paragraph{(i) Bounding $1 - \| \PQp \PTO \|^2$.} We will apply Lemma \ref{lem:op-appx} with $S = T + \Omega$. The lemma requires $$\dim{Q} \;\ge\; C_1 \cdot  \dim{T+\Omega}.$$ The dimension of $T + \Omega$ is a random variable, which depends on the size of the support set $\Omega$. Let $\event_\Omega$ denote the event
\begin{equation}
\event_{\Omega} = \set{ | \Omega | \;\le\; 2 \rho mn + m }. 
\end{equation}
Notice that $|\Omega|$ is a sum of $mn$ $\bernoulli{\rho}$ random variables. By Bernstein's inequality, 
\begin{equation}
\prob{ |\Omega| \ge \rho m n + t } \;\le\; \exp\left( \frac{-t^2 / 2}{ \rho m n + t / 3 } \right).
\end{equation}
Setting $t = \rho m n + m$ and simplifying, we obtain
\begin{equation}
\prob{\event_\Omega^c} \le \exp\left( - \frac{3 m}{10} \right).
\end{equation}
On $\event_{\Omega}$, we have 
\begin{equation} \label{eqn:OT-dim-bound}
\dim{\Omega + T} \;<\; 2 \rho mn + m + 2 mr \;\le\; 3 \cdot ( \, \rho mn + m r\, ). 
\end{equation}
Comparing \eqref{eqn:OT-dim-bound} to the condition on $\dim{Q}$ in Theorem \ref{thm:main}, we can see that on $\event_{\Omega}$, the conditions of Lemma \ref{lem:op-appx} are satisfied. Now, let $S = T + \Omega$, set $B = \{ \mb X \in S \mid \| \mb X \|_F = 1 \}$ and notice that 
\begin{eqnarray}
1 - \| \PQp \PS \|^2 &=& \inf_{\mb X \in B} \;\; \< \mb X, \mb X \> - \< \PQp \PS \mb X, \PQp \PS \mb X \> \nonumber \\
&=& \inf_{\mb X \in B}\;\; \left\< \mb X, \left( \PS - \PS \PQp \PS \right) \mb X \right\> \nonumber \\
&=& \inf_{\mb X \in B} \;\; \left\< \mb X, \left( \frac{\dim{Q}}{mn} \PS + \PS \PQ \PS - \frac{\dim{Q}}{mn} \PS  \right) \mb X \right\> \nonumber \\
&\ge& \frac{\dim{Q}}{mn} - \sup_{\mb X \in B} \; \left\< \mb X, \left( \PS \PQ \PS - \frac{\dim{Q}}{mn} \PS  \right) \mb X \right\> \nonumber \\
&\ge& \frac{\dim{Q}}{mn} - \left\| \PS \PQ \PS - \frac{\dim{Q}}{mn} \PS \right\|.
\end{eqnarray}
Let $\event_Q$ be the event $\left\{ \| \PS \PQ \PS - \tfrac{\dim{Q}}{mn} \PS \| \le \tfrac{1}{16} \tfrac{\dim{Q}}{mn} \right\}$. Using Lemma \ref{lem:op-appx} and $\dim{S} = \dim{T \oplus \Omega} \ge m$, we have
\begin{equation}
\condprob{\event_Q}{\event_\Omega} \;\ge\; 1 - C_2 \exp( - c_1 m ).
\end{equation}
On $\event_Q$, 
\begin{equation*} 
1 - \| \PQp \PS \|^2 \;\ge\; \frac{15}{16} \frac{\dim{Q}}{mn}.
\end{equation*}
Since by assumption $\dim{Q} \ge C_Q \times \log^2 m \times \dim{T+\Omega} \ge C_Q \times \log^2 m \times m$ and $m \ge n$, ensuring that $C_Q > 16 / 15$, we can further conclude that 
\begin{equation} \label{eqn:pqppslb}
1 - \| \PQp \PS \|^2 \;\ge\; \frac{1}{m}.
\end{equation}

\paragraph{(ii) Inexact PCP Certificate.} By Theorem \ref{thm:clmw}, on an event $\event_{\mathrm{PCP}}$ of probability at least $1 - C_2 m^{-10}$, we have $\| \PT \PO \| < 1/2$, and there exists an $(m^{-2},1/4)$-inexact PCP certificate $\cert{\mathrm{PCP}}$ for $(\mb L_0,\mb S_0)$, with 
\begin{equation} \label{eqn:lambda-pcp-bound}
\| \cert{\mathrm{PCP}}\|_F \;\le\; C_3 \sqrt{\mathrm{rank}(\mb L_0)} + 2 \lambda \sqrt{ |\Omega| }.
\end{equation}
Moreover, $\event_{\mathrm{PCP}}$ is independent of $Q$. We rewrite the bound \eqref{eqn:lambda-pcp-bound} a bit for later use. We have 
\begin{eqnarray*}
m \, \| \cert{\mathrm{PCP}}\|_F^2 \;\le\; m  \left( \, 2 \, C_3^2 \, r + 4 \lambda^2 |\Omega| \, \right),
\end{eqnarray*}
which on $\event_\Omega$ gives 
\begin{equation} \label{eqn:cert-sqr-bound}
m \, \| \cert{\mathrm{PCP}}\|_F^2 \;\le\; C_4 \,( \,\rho mn + m r\, ),
\end{equation}
where $C_4$ is numerical. 

\paragraph{(iii) Upgrade to CPCP Certificate.} Now, condition on $\event_{\mathrm{PCP}}$ and $\event_\Omega$. By our assumption on $\dim{Q}$, the conditions of Theorem \ref{thm:upgrade} are satisfied. On an event $\event_{\text{upgrade}}$ of conditional probability at least $1 - C_5 m^{-9}$, the certificate $\cert{\mathrm{PCP}}$ can be refined to an $(\alpha',\beta')$-inexact CPCP certificate $\cert{\mathrm{CPCP}}$, with 
\begin{equation}
\alpha' \le m^{-2} + m^{-3} \| \cert{\mathrm{PCP}} \|_F,
\end{equation}
and 
\begin{equation*}
\beta' \;\le\; \frac{1}{4} + C_6 \left[ \frac{\| \mb \Lambda_{\mathrm{PCP}} \|_F^2 \log m}{\dim{Q}} \right]^{1/2}  \max\set{\expect{ \,\| \mb G \|\,}  + \sqrt{\log m}\;, \; \expect{ \, \| \mb G \|_\infty } \sqrt{m} + \sqrt{m \log m} },
\end{equation*}
where $\mb G$ is iid $\normal{0}{1}$. Furthermore, provided $mn > 1$, we have the bounds 
\begin{equation}
\E\| \mb G \| \;\le\; 2 \sqrt{m} \quad \text{and} \quad \E \| \mb G \|_\infty \;\le\; 3 \sqrt{2 \log m},
\end{equation}
and so 
\begin{eqnarray*}
\beta' &\le& \frac{1}{4} + C_7 \left[ \frac{ m \, \| \mb \Lambda_{\mathrm{PCP}} \|_F^2 \, \log^2 m }{\dim{Q}} \right]^{1/2} \quad\le\quad \frac{1}{4} + \left[ \frac{ C_8 ( \rho m n + mr ) \log^2 m }{\dim{Q}} \right]^{1/2},
\end{eqnarray*}
where we have used \eqref{eqn:cert-sqr-bound}. Ensuring that the constant $C_Q$ in the statement of the theorem is larger than $16 C_8$, we can conclude that $\beta' \le 1/2$.

Referring to property (II) above, all that is left to show is that $\alpha' < \frac{1-\| \PQp\PTO\|_F^2}{4 \sqrt{m}}$. Using paragraph (i), on $\event_\Omega \cap \event_Q$, it suffices to show $\alpha' < \frac{1}{4 m^{3/2}}$. Using \eqref{eqn:cert-sqr-bound}, ensuring that the constants $c_r,c_\rho$ in the statement of Theorem \ref{thm:main} are sufficiently small (say, each smaller than $1/2 C_4$), we may conclude that $\| \mb \Lambda_{\mathrm{PCP}} \| \le \sqrt{m}$. Hence, we have $\alpha' \le m^{-2} + m^{-5/2}$, which is strictly smaller than $\frac{1}{4 m^{3/2}}$ provided $m$ is sufficiently large. 

We have shown that on 
$$\event_{\text{good}} \doteq \event_\Omega \cap \event_{Q} \cap \event_{\text{PCP}} \cap \event_{\text{upgrade}},$$
(I)-(II) hold, and hence $(\mb L_0,\mb S_0)$ is the unique optimal solution to the CPCP problem. 

\paragraph{(iv) Probability.} We have 
\begin{eqnarray*}
\prob{ \event_{\text{good}}^c } &\le& \prob{ ( \event_Q \cap \event_\Omega )^c } + \prob{ (\event_{\text{upgrade}} \cap \event_{\mathrm{PCP}} \cap \event_\Omega)^c } \\
&=& 1 - \condprob{ \event_Q }{\event_{\Omega}} \prob{\event_{\Omega}} + 1 - \condprob{\event_{\text{upgrade}}}{ \event_{\mathrm{PCP}} \cap \event_\Omega} \prob{\event_{\mathrm{PCP}} \cap \event_{\Omega} } \\
&\le& 1 - \condprob{ \event_Q }{\event_{\Omega}} + \prob{\event_{\Omega}^c} + 1 - \condprob{\event_{\text{upgrade}}}{ \event_{\mathrm{PCP}} \cap \event_\Omega} + \prob{\event_{\mathrm{PCP}}^c } + \prob{ \event_{\Omega}^c } \\
&\le& C_2 \exp( - c_1 m )  + C_2 m^{-10} + C_5 m^{-9} + 2 \exp( - 3 m / 10 ),
\end{eqnarray*}
provided that $m$ is larger than some $m_0$. Consolidating bounds, we may conclude that correct recovery occurs with probability at least $1 - C_9 m^{-9}$, choosing $C_9$ such that the bound is nontrivial only for $m > m_0$. This completes the proof of Theorem \ref{thm:main}. 
\end{proof}

We close by sketching the proof of Proposition \ref{thm:clmw}:

\begin{proof}[\bf \em Proof of Proposition \ref{thm:clmw} (sketch)] Under the hypotheses, the bound $\| \PO \PT \| \le 1/2$ follows immediately from Corollary 2.7 of \cite{Candes2011-JACM}. \cite{Candes2011-JACM} constructs a certificate $\Wh$ in three parts as 
\begin{equation}
\Lpcp = \mb U \mb V^* + \mb W^L + \mb W^S.
\end{equation} 
The two terms $\mb W^L$, $\mb W^S$ will both be elements of $T^\perp$, so 
\begin{equation}
\|\PT\Lpcp - \mb U \mb V^* \|_F = 0.
\end{equation}
Moreover, the term $\mb W^S$ will satisfy $\PO \mb W^S = \lambda \, \mathrm{sign}(\mb S_0)$. So, we have 
\begin{equation}
\|\PO \Lpcp - \lambda \, \mathrm{sign}(\mb S_0) \|_F = \| \PO[ \mb U \mb V^* + \mb W^L ]\|_F.
\end{equation}
We can therefore take $\alpha = \|\PO[ \mb U \mb V^* + \mb W^L]\|_F$. To prove Proposition \ref{thm:clmw}, it is therefore enough to show that with high probability the following properties are satisfied:
\begin{itemize}
\item{\bf (I) Structure constraint:} $\| \PO [ \mb U \mb V^* + \mb W^L ] \|_F \le m^{-2}$. 

\item{\bf (II) Dual norm constraints:}  
\begin{equation} \label{eqn:WL-dual}
\| \PTp \mb W^L \| \le 1/8, \qquad \| \POc [ \mb U \mb V^* + \mb W^L ] \|_\infty \;\le\; \frac{\lambda}{8},
\end{equation}
and 
\begin{equation} \label{eqn:WS-dual}
\| \PTp \mb W^S \| \le 1/8, \qquad \| \PO^c \mb W^S \|_\infty \;\le\; \frac{\lambda}{8}.
\end{equation}

\item{\bf (III) Frobenius norm bounds:} We have $\| \Lpcp \|_F  \le \| \mb U \mb V^* \|_F + \| \mb W^L \|_F + \| \mb W^S \|_F$. The first term is simply $\sqrt{r}$. We will show that 
\begin{eqnarray}
\|\mb W^L\|_F &\le& 3 \sqrt{r}, \label{eqn:WL-fro} \\
\|\mb W^S\|_F &\le& \frac{4}{3} \lambda \sqrt{\| \mb S_0 \|_0}. \label{eqn:WS-fro}
\end{eqnarray}
\end{itemize}

In paragraph (i) below, we review the construction of $\mb W^L$ from \cite{Candes2011-JACM}, and show that with slight changes in the constants, (I) and \eqref{eqn:WL-dual} are satisfied with high probability. In paragraph (ii) we review the construction of $\mb W^S$ and show that with high probability \eqref{eqn:WS-dual} is satisfied. Together, this implies that $\Lpcp$ is an $(m^{-2},1/4)$-certificate for the PCP problem. Paragraph (ii) will also show \eqref{eqn:WS-fro}. Finally, in paragraph (iii), we show \eqref{eqn:WL-fro}. This step involves the most additional work. Taken together, this establishes the proposition.

\paragraph{(i) Constructing $\mb W^L$.} The term $\mb W^L$ is constructed to lie in $T^\perp$ and satisfy $$\PO[ \, \mb U \mb V^* + \mb W^L \, ]\approx \mb 0.$$ 
This is accomplished via a golfing argument that writes the complement $\Omega^c$ as a union of $j_0$ subsets $\Upsilon_1, \dots, \Upsilon_{j_0}$, with $\Upsilon_j \sim_{iid} \bernoulli{q}$.\footnote{In \cite{Candes2011-JACM}, $\Upsilon_j$ is denoted $\Omega_j$; we use the notation $\Upsilon$ to avoid confusion between the support $\Omega$ and the subsets $\Upsilon$ of the complement of the support.} The parameter $q$ is set so that $\rho = (1-q)^{j_0}$, which ensures that $\Omega$ is indeed $\bernoulli{\rho}$. Notice that with this setting, we have $q \ge (1-\rho)/j_0$. 

The certificate is generated inductively, starting with $\mb Y_0 = \mb 0$. The error at step $j$ is $$\mb Z_j = \PT \mb Y_j - \UVt$$ and the corrective update $$\mb Y_j = \mb Y_{j-1} - q^{-1} \mc P_{\Upsilon_j} \mb Z_{j-1}.$$
This leads to a recursive expression for the error $\mb Z_{j}$:
\begin{equation}
\mb Z_{j} = \PT ( \mc I - q^{-1} \mc P_{\Upsilon_j} ) \PT \mb Z_{j-1},
\end{equation}
which implies that the error decays quickly in both $\ell^\infty$ and Frobenius norm:
\begin{eqnarray*} 
\| \mb Z_j \|_\infty &\le& \tfrac{1}{2}\| \mb Z_{j-1} \|_\infty \;\le\; 2^{-j} \| \mb Z_0 \|_\infty, \\
\| \mb Z_j \|_F &\le& \tfrac{1}{2} \| \mb Z_{j-1} \|_F \;\le\; 2^{-j} \| \mb Z_0 \|_F. 
\end{eqnarray*}
These inequalities are \cite{Candes2011-JACM} (3.4)-(3.6), with $\epsilon = 1/2$. Notice that \cite{Candes2011-JACM} (3.3) is indeed satisfied, provided $c_r > 0$ is sufficiently small. Using 
\begin{equation}
\| \mb Z_0 \|_\infty \le \sqrt{\frac{\mu r}{ m n}}, \quad\text{and}\quad \| \mb Z_0 \|_F = \|\mb U \mb V^* \|_F = \sqrt{r},
\end{equation}
 we obtain
\begin{eqnarray} 
\sum_{j=0}^{j_0} \| \mb Z_j \|_\infty \;\le\; 2 \sqrt{\frac{\mu r}{mn}}, \quad \text{and} \quad \sum_{j=0}^{j_0} \| \mb Z_j \|_F \;\le\; 2 \sqrt{r}.                     \label{eqn:Z-sum-bounds}
\end{eqnarray}
From arguments of \cite{Candes2011-JACM}, these bounds hold simultaneously on an event $\event_{Z}$ of probability at least $1 - C_1 m^{-10}$. After $j_0$ steps, the component $\mb W^L$ is generated as
\begin{equation} \label{eqn:WL}
\mb W^L \;\doteq\; \PTp \mb Y_{j_0} \;=\; \PTp \sum_{j=1}^{j_0} q^{-1} \mc P_{\Upsilon_j} \mb Z_{j-1}.
\end{equation}
The proof of Lemma 2.8(b) of \cite{Candes2011-JACM} shows that 
\begin{equation}
\| \PO [ \mb U \mb V^* + \mb W^L ] \|_F \quad\le\quad \| \mb Z_{j_0}\|_F.   
\end{equation}
In \cite{Candes2011-JACM}, $j_0$ was chosen to ensure that $\| \mb Z_{j_0}\|_F \le 1 / 4 \sqrt{m}$. Here, we set $j_0 = \ceiling{3 \log_2 m}$, ensuring that $\| \mb Z_{j_0} \|_F \le 2^{-j_0} \sqrt{r} \le m^{-2}$. The arguments of \cite{Candes2011-JACM} establish the following: 
\begin{eqnarray}
\| \PTp \mb W^L \| &\le& 2 C_0' \sqrt{\frac{m\log m}{q}} \| \mb U \mb V^* \|_\infty, \\
\| \POc \mb W^L \|_\infty &\le& \| \mb Z_{j_0} \|_\infty + 2 q^{-1} \| \mb U \mb V^* \|_\infty.
\end{eqnarray}
where $C_0'$ is numerical. Moreover, we know that $q > \frac{c}{\log m}$, where $c > 0$ is numerical. Hence, we have 
\begin{eqnarray}
\| \PTp \mb W^L \| &\le& C \sqrt{\frac{ \mu r \log^2 m}{n} }, \\
\| \POc \mb W^L \|_\infty &\le& \frac{1}{\sqrt{m}} \left( \sqrt{\frac{\mu r}{n}} + \frac{2}{c} \sqrt{\frac{\mu r  \log^2 m}{n} } \right).
\end{eqnarray}
Recalling that by assumption $r \le c_r n / \mu \log^2 m$, we have the desired bounds \eqref{eqn:WL-dual}, provided the constant $c_r$ is sufficiently small. 

\paragraph{(ii) Constructing $\mb W^S$.} The term $\mb W^S$ is constructed to satisfy $\mb W^S \in T^\perp$, $\PO \mb W^S = \lambda \,\mathrm{sign}(\mb S_0)$. More precisely, \cite{Candes2011-JACM} set this term to be the minimum Frobenius norm solution to the system of equations 
\begin{equation}
\PT \mb W^S = \mb 0, \quad \PO \mb W^S = \lambda \, \mathrm{sign}(\mb S_0). 
\end{equation}
This solution is given by the Neumann series 
\begin{equation}
\mb W^S \;=\; \PTp \sum_{j=0}^\infty ( \PO \PT \PO )^j [ \lambda \, \mathrm{sign}(\mb S_0) ].
\end{equation}
So,
\begin{equation}
\| \mb W^S \|_F \;\le\; \frac{\| \lambda \, \mathrm{sign}(\mb S_0) \|_F}{1-\| \PO \PT \|^2} \;\le\; \tfrac{4}{3} \lambda \sqrt{\| \mb S_0 \|_0}.
\end{equation}
Similar to paragraph (i) above, one can quickly check that by ensuring $\rho$ is smaller than some fixed constant and using the same arguments as \cite{Candes2011-JACM}, \eqref{eqn:WS-dual} is satisfied with high probability. 

\paragraph{(iii) Bounding $\| \mb W^L \|_F$.} We use the fact that $\Upsilon_j$ and $\mb Z_{j-1}$ are independent random variables. By \eqref{eqn:WL}, it is enough to control the Frobenius norm  $q^{-1} \mc P_{\Upsilon_j} \mb Z_{j-1}$ for each $j$. Notice that
\begin{eqnarray*}
\|  q^{-1} \mc P_{\Upsilon_j} \mb Z_{j-1} \|_F^2 &=& \| \mb Z_{j-1} \|_F^2 + \sum_{kl} (q^{-1} \delta_{kl}-1) [ \mb Z_{j-1} ]_{kl}^2 \quad\doteq\quad \| \mb Z_{j-1}\|_F^2 + \sum_{kl} H_{kl},
\end{eqnarray*}
where $\delta_{kl}$ is an indicator for the event $(k,l) \in \Upsilon_j$. Then $\E[ H_{kl} ] = 0$, $|H_{kl}| \le q^{-1} \| \mb Z_{j-1} \|_\infty^2$ almost surely, and $\E [ H_{kl}^2 ] \le q^{-1} [\mb Z_{j-1}]_{kl}^4$. Summing, we have 
\begin{equation}
\sum_{kl} \E[ H_{kl}^2 ] \;\le\; q^{-1} \| \mb Z_{j-1} \|_\infty^2 \| \mb Z_{j-1} \|_F^2. 
\end{equation}
By Bernstein's inequality, 
\begin{eqnarray*}
\prob{ \| q^{-1} \mc P_{\Upsilon_j} \mb Z_{j-1} \|_F^2 \;>\; \| \mb Z_{j-1} \|_F^2 + t } &\le& \exp\left( - \frac{t^2}{2 q^{-1} \| \mb Z_{j-1} \|_\infty^2  \left( \| \mb Z_{j-1} \|_F^2 + \tfrac{t}{3} \right)} \right).
\end{eqnarray*}
By setting
\begin{equation}
t_j = C_2 \max\set{ \| \mb Z_{j-1} \|_{\infty}^2 q^{-1} \log m\,, \;\| \mb Z_{j-1} \|_{\infty} \|\mb Z_{j-1} \|_F \sqrt{q^{-1}\log m} },
\end{equation}
with appropriate numerical constant $C_2$, we can ensure that for each $j$, 
\begin{equation}
\prob{ \| q^{-1} \mc P_{\Upsilon_j} \mb Z_{j-1} \|_F^2 \;>\; \| \mb Z_{j-1} \|_F^2 + t_j } \;\le\; m^{-11}.
\end{equation}
Since we have $q > \frac{c}{\log m}$ for some positive numerical constant $c$, using $\sqrt{s+t} \le \sqrt{s}+\sqrt{t}$, on an event with overall probability at least $1-j_0 m^{-11}$, 
\begin{eqnarray*}
\| \mb W^L \|_F &\le& \sum_{j=1}^{j_0} \| q^{-1} \mc P_{\Upsilon_j} \mb Z_{j-1} \|_F \\
                &\le& \sum_{j=1}^{j_0} \| \mb Z_{j-1} \|_F + \sqrt{t_j} \\
                &\le& \sum_{j=1}^{j_0} \| \mb Z_{j-1} \|_F + C_3  \| \mb Z_{j-1}\|_\infty \log m + C_4 \sqrt{\| \mb Z_{j-1} \|_\infty \| \mb Z_{j-1} \|_F \log m} \\
                &\le& 2 \sqrt{r} + 2 C_3 \log(m) \sqrt{\frac{\mu r}{mn}} + C_4 \sqrt{\log m} \times \sum_{j=1}^{j_0} 2^{-j} \| \mb Z_{0} \|_\infty^{1/2} \| \mb Z_{0} \|_F^{1/2} \\
                &\le& 2 \sqrt{r} + 2 C_3 \sqrt{\frac{1}{m}\frac{\mu r \log^2 m}{n}} + 2 C_4 \sqrt[4]{\frac{r}{m} \frac{\mu r \log^2 m}{n}}.
\end{eqnarray*}
Recalling again the assumption $r \le c_r n / \mu \log^2 m$, and ensuring that $c_r$ is sufficiently small, the final two terms above are bounded by constants. In particular, we can conclude that $\| \mb W^L \|_F \le 3 \sqrt{r}$. This completes the proof. 
\end{proof}

\section*{Acknowledgements} JW thanks the Rice group (Andrew Waters, Aswin Sankaranarayanan and Richard Baraniuk) for discussions and clarifications related to this work and \cite{Waters2011-NIPS}. He would also like to thank Xiaodong Li of Stanford for discussions related to this work. 

\bibliographystyle{alpha}
\bibliography{gpcp}

\appendix

\section{Proof of Lemma \ref{lem:optimality}: Optimality Conditions} \label{app:gen-duality}

\begin{proof} 
Let $f$ denote the objective function. Consider a feasible perturbation $\mb \delta = (\mb \Delta_{1}, \dots, \mb \Delta_\tau)$, so $\PQ\sum_i \mb \Delta_i = \mb 0$. Then for any $\mb W_1, \dots, \mb W_\tau$ such that $\forall \,i, \;\, \mb W_i \in \partial \| \cdot \|_{(i)}(\mb X_{i,\star})$, we have 
\begin{equation} \label{eqn:subgrad-bound}
f(\mb x_\star + \mb \delta) \;\ge\; f(\mb x_\star) + \sum_i \lambda_i \< \mb W_i, \mb \Delta_i \>.
\end{equation}
By duality of norms, for each $i$ there exists $\mb H_i \in \Re^{m \times n}$ with $\| \mb H_i \|_{(i)}^* \le 1$ and 
\begin{equation}
\<\mb H_i , \PTip \mb \Delta_i \> \;=\; \| \PTip \mb \Delta_i \|_{(i)}.
\end{equation}
Set $\mb W_i = \mb S_i + \PTip \mb H_i$. From our definition of a decomposable norm, $\PTip$ is nonexpansive, and so $\| \PTip \mb H_i \|_{(i)}^* \le 1$, and $\mb W_i \in \partial \| \cdot \|_{(i)}( \mb X_{i,\star} )$. Moreover, 
\begin{eqnarray}
\< \mb W_i, \mb \Delta_i \> &=& \< \PTi \mb W_i,\mb \Delta_i \> + \< \PTip \mb W_i, \mb \Delta_i \> \nonumber \\
&=& \< \PTi \mb W_i , \PTi \mb \Delta_i \> + \< \PTip \mb W_i, \PTip \mb \Delta_i \> \nonumber \\
&=& \< \mb S_i, \PTi \mb \Delta_i \> + \< \mb H_i, \PTip \mb \Delta_i \> \nonumber \\
&=& \< \mb S_i, \PTi \mb \Delta_i \> + \| \PTip \mb \Delta_i \|_{(i)}
\end{eqnarray}
Plugging in to \eqref{eqn:subgrad-bound}, we have
\begin{eqnarray}
f(\mb x_\star + \mb \delta) &\ge& f(\mb x_\star) + \sum_i \< \lambda_i \mb S_i, \mc P_{T_i} \mb \Delta_i \> + \lambda_i \| \mc P_{T_i^\perp} \mb \Delta_i \|_{(i)} \nonumber \\
&=& f(\mb x_\star) + \sum_i \< \PTi \mb \Lambda, \mc P_{T_i} \mb \Delta_i \> + \lambda_i \| \mc P_{T_i^\perp} \mb \Delta_i \|_{(i)} \nonumber \\
&=& f(\mb x_\star) + \sum_i \< \mb \Lambda, \mc P_{T_i} \mb \Delta_i \> + \lambda_i \| \mc P_{T_i^\perp} \mb \Delta_i \|_{(i)} \nonumber \\
&=& f(\mb x_\star) + \sum_i \< \mb \Lambda, \mb \Delta_i \> - \< \mb \Lambda, \mc P_{T^\perp_i} \mb \Delta_i \> + \lambda_i \| \mc P_{T_i^\perp} \mb \Delta_i \|_{(i)} \nonumber \\
&\ge& f(\mb x_\star) + \Bigl\< \mb \Lambda, \sum_j \mb \Delta_j \Bigr\> + \sum_i - \| \mc P_{T^\perp_i} \mb \Lambda \|_{(i)}^* \| \mc P_{T^\perp_i} \mb \Delta_i \|_{(i)} + \lambda_i \| \mc P_{T_i^\perp} \mb \Delta_i \|_{(i)} \nonumber \\
&=& f(\mb x_\star) + \Bigl\< \PQp \mb \Lambda, \sum_j \mb \Delta_j \Bigr\> + \sum_i \left(\lambda_i - \| \mc P_{T^\perp_i} \mb \Lambda \|_{(i)}^* \right) \| \mc P_{T_i^\perp} \mb \Delta_i \|_{(i)} \nonumber \\
&=& f(\mb x_\star) + \sum_i \left(\lambda_i - \| \mc P_{T^\perp_i} \mb \Lambda \|_{(i)}^* \right) \| \mc P_{T_i^\perp} \mb \Delta_i \|_{(i)},
\end{eqnarray}
where we have used that $\< \PQ \mb \Lambda, \sum_j \mb \Delta_j \> = \< \mb \Lambda, \PQ \sum_j \mb \Delta_j \> = 0$, since $\mb \delta$ is feasible. Since each of the $\| \mc P_{T^\perp_i} \mb \Lambda \|_{(i)}^*$ is strictly smaller than $\lambda_i$, if any of the $\mc P_{T_i^\perp} \mb \Delta_i$ are nonzero, then $f(\mb x_\star + \mb \delta) > f(\mb x_\star)$. If, on the other hand, all of the $\mc P_{T_i^\perp} \mb \Delta_i$ are zero, then $\mb \Delta_i \in T_i$ for all $i$, and the constraint $\PQ \sum_i \mb \Delta_i = \mb 0$ implies that $\sum_i \mb \Delta_i \in (T_1 + \dots + T_\tau) \cap Q^\perp$. If $\sum_i \mb \Delta_i \ne \mb 0$, this contradicts independence of $(T_1, \dots, T_\tau,Q^\perp)$. If $\sum_i \mb \Delta_i = \mb 0$, this contradicts independence of $T_1, \dots, T_\tau$ (which follows from independence of $(T_1, \dots, T_\tau, Q^\perp)$). So, we conclude that for any feasible perturbation $\mb \delta$, $f(\mb x_\star + \mb \delta)$ is strictly larger than $f(\mb  x_\star)$. 
\end{proof}

\section{Proof of Lemma \ref{lem:op-appx}: Operator Approximations} \label{app:op-appx}

\newcommand{\PVs}{\mb P_{\mathrm{vec}[S]}}

\begin{proof} Fix an $1/4$-net $\Gamma$ for the unit ball restricted to $S$. By \cite{Ledoux} Proposition 4.16, there exists such a net of size at most $\exp( \dim{S} \log 12 )$. Let $\mc H : \Re^\gamma \to \Re^{m \times n}$ via $\mc H \mb x = \sum_{i = 1}^\gamma \mb H_i x_i$, and let $\psi : \Re^\gamma \to \Re^{m \times n}$ via $\psi \mb x = \sum_{i=1}^\gamma \bar{\mb H}_i x_i$, where $(\bar{\mb H}_i)$ is an orthonormal sequence of matrices that span $R$. By the Bartlett decomposition, we may assume that $\left[ \vec{\bar{\mb H_1}} \mid \dots \mid \vec{\bar{\mb H_\gamma}} \right] \in \Re^{mn \times \gamma}$ is distributed according to the Haar measure on the Stiefel manifold of $mn \times \gamma$ matrices with orthonormal columns. Moreover, we have $\mc A = \mc H \mc H^*$ and $\mc P_R = \psi \psi^*$. 

A standard argument (see \cite{Vershynin2011} Lemma 5.4) gives that 
\begin{eqnarray}
\norm{ \PS \partop{} \PS - \PS }{} \quad= \sup_{\begin{array}{r}\mb X \in S \\ \| \mb X \|_F = 1\end{array}} \magnitude{ \frac{mn}{\gamma} \| \mc H^* \mb X \|_2^2  - 1 } \quad\le\quad 2 \sup_{\mb X \in \Gamma} \magnitude{ \frac{mn}{\gamma} \| \mc H^* \mb X \|_2^2  - 1 }.
\end{eqnarray}
Notice that $\sqrt{\frac{mn}{\gamma}} \mc H^* \mb X$ is distributed as an iid $\mc N(0,1/\gamma)$ random vector. Using Lemma 1 of \cite{Laurent2000-AOS}, 
\begin{equation}
\prob{ \magnitude{ \frac{mn}{\gamma} \| \mc H^* \mb X \|_2^2  - 1 } \ge 2 \sqrt{\frac{t}{\gamma}} + 2 \frac{t}{\gamma} } \;\le\; 2 \, e^{-t}. 
\end{equation}
Choose $t = c_1 \gamma$, with $c_1$ small enough that $4 \sqrt{c_1} + 4 c_1 \le 1/2$. Take a union bound over all $\exp(
 \dim{S} \log 12 )$ elements of $\Gamma$ to get 
 \begin{equation}
\prob{ \norm{ \PS \partop{} \PS - \PS }{} \ge \frac{1}{2} } \;\le\; 2 \exp\left(-c_1 \gamma + \dim{S} \log 12\right).
\end{equation}
Using the assumption that $\gamma > C_1 \dim{S}$, and ensuring that $C_1$ is large enough that $c_1 > \frac{\log 12}{C_1}$ completes the proof of \eqref{eqn:first}.

For the second term, we repeat the argument, noting that 
\begin{eqnarray}
\norm{ \frac{mn}{\gamma} \PS \mc P_R \PS - \PS }{} &\le& 2 \sup_{\mb X \in \Gamma} \magnitude{ \frac{mn}{\gamma} \| \psi^* \mb X \|_2^2  - 1 }.
\end{eqnarray}
Note that $\| \psi^* \mb X \|_F^2 = \| (\mathrm{vec} \circ \psi)^* \vec{\mb X} \|_2^2$. The operator $\mathrm{vec} \circ \psi : \Re^{\gamma} \to \Re^{mn}$ can be identified with an $mn \times \gamma$ matrix $\mb U$, which per the above discussion can be taken to be distributed according to the Haar measure. By orthogonal invariance, for any fixed $\mb x$, $\mb U^* \mb x$ is equal in distribution to the restriction of uniformly distributed random unit vector $\mb r \in \bb S^{mn-1}$ to its first $\gamma$ coordinates. Lemma 2.2 of \cite{Dasgupta2003-RSA} provides convenient tail bounds for the norm of such a coordinate restriction. Applying that lemma, we have that for every $t > 0$, there exists $c_t > 0$ such that 
\begin{equation}
\prob{ \magnitude{ \frac{mn}{\gamma} \| \psi^* \mb X \|_2^2  - 1 } > t } \;\le\; \exp\left( - c_t \gamma \right).
\end{equation}
Set $t = 1/32$. As above, ensuring that $C_1$ is larger than $\frac{\log 12}{c_t}$ and taking a union bound shows that with the desired probability $\norm{ \frac{mn}{\gamma} \PS \mc P_R \PS - \PS }{} \le 1/16$. Rescaling gives the bound quoted in the statement of the lemma. 
\end{proof}

\end{document}

%% file: jw_defs.tex
\usepackage[pdftex]{graphicx,color}
\usepackage{amsmath}
\usepackage{amssymb}
\usepackage{graphicx}
\usepackage{color}
\usepackage{url}
\usepackage{amsfonts, amscd, amsthm}
\usepackage{algorithm}
\usepackage{algorithmic}
\usepackage{subfigure}

\renewcommand{\Re}{\mathbb R}

\newcommand{\supp}[1]{\mathrm{supp}\left(#1\right)}
\newcommand{\rank}{\mathrm{rank}}
\newcommand{\<}{\langle}
\renewcommand{\>}{\rangle}

\renewcommand{\vec}[1]{\mathrm{vec}\left[ #1 \right]}

\renewcommand{\mathbf}{\boldsymbol}
\newcommand{\event}{\mathcal{E}}
\renewcommand{\P}{\mathbb{P}}
\newcommand{\E}{\mathbb{E}}

\newcommand{\PO}{\mathcal{P}_\Omega}
\newcommand{\POc}{\mathcal{P}_{\Omega^c}}

\newcommand{\bb}{\mathbb}
\newcommand{\mb}{\mathbf}
\newcommand{\mc}{\mathcal}

\newcommand{\minimize}[2]{\text{minimize} \quad #1 \quad \text{subject to} \quad #2}
\newcommand{\prob}[1]{\P\left[ \, #1 \, \right] }
\newcommand{\condprob}[2]{\P\left[ \, #1 \mid #2 \, \right] }

\renewcommand{\dim}[1]{\mathrm{dim}(#1)}
\newcommand{\norm}[2]{\left\| #1 \right\|_{#2}}
\newcommand{\expect}[1]{\E\left[ #1 \right]}
\newcommand{\condexp}[2]{\E\left[ #1 \mid #2 \right]}
\newcommand{\innerprod}[2]{\left\< #1, #2 \right\>}
\newcommand{\magnitude}[1]{\left|#1\right|}
\newcommand{\sign}[1]{\mathrm{sign}\left(#1\right)}
\newcommand{\set}[1]{\left\{ #1 \right\}}
\renewcommand{\rank}[1]{\mathrm{rank}\left(#1\right)}

\newcommand{\PQp}{\mc P_{\Qp}}
\newcommand{\PT}{\mc P_T}
\newcommand{\PQ}{\mc P_Q}
\newcommand{\PTO}{\mc P_{T \oplus \Omega}}

\newcommand{\PTp}{\mc P_{T^\perp}}

\newcommand{\PS}{\mc P_S}

\newcommand{\PSp}{\mc P_{S^\perp}}

\newcommand{\Qp}{Q^\perp}

\newcommand{\Wh}{\hat{\mb \Lambda}}
\newcommand{\contract}[1]{\left( \mc I - \frac{mn}{\gamma} \mc A_{#1} \right)}
\newcommand{\partop}[1]{\frac{mn}{\gamma} \mc A_{#1}}

\newcommand{\cert}[1]{\mb \Lambda_{#1}}
\newcommand{\err}[1]{\mb E_{#1}}
\newcommand{\UVt}{\mb U \mb V^*}
\newcommand{\expfrac}[2]{\exp\left(-\frac{#1}{#2}\right)}

\newcommand{\Lh}{\hat{\mb \Lambda}}
\newcommand{\PTi}{\mc P_{T_i}}
\newcommand{\PTip}{\mc P_{T_i^\perp}}

\newcommand{\bernoulli}[1]{\mathrm{Ber}(#1)}
\newcommand{\normal}[2]{\mc N \left( #1 , #2 \right) }

\newcommand{\eqdist}{\equiv_d}
\newcommand{\skewvariance}{\sum_{l=1}^\gamma \innerprod{ \mb H'_l }{\PS \mb M }^2}
\newcommand{\skewstdev}{\left( \skewvariance \right)^{1/2}}

\newcommand{\Lpcp}{\mb \Lambda_{\mathrm{PCP}}}

\newcommand{\floor}[1]{\left\lfloor #1 \right\rfloor}
\newcommand{\ceiling}[1]{\left\lceil #1 \right\rceil}

\newtheorem{theorem}{Theorem}[section]

\newtheorem{lemma}[theorem]{Lemma}
\newtheorem{proposition}[theorem]{Proposition}

\newtheorem{remark}[theorem]{Remark}

\newtheorem{definition}[theorem]{Definition}